\documentclass{giq19my}

\newtheorem{defi}{Definition}
\theoremstyle{definition}
\newtheorem{ex}{Example}

\def\cl{{C}\!\ell}
\def\j{{\rm j}}
\def\k{{\rm k}}
\def\ii{{\rm i}}
\def\Mat{{\rm Mat}}
\def\R{{\mathbb R}}
\def\C{{\mathbb C}}
\def\H{{\mathbb H}}
\def\Z{{\mathbb Z}}
\def\FF{{\mathbb F}}
\def\F{{\rm F}}
\def\M{{\rm M}}
\def\G{{\rm G}}
\def\U{{\rm U}}
\def\SU{{\rm SU}}
\def\SL{{\rm SL}}
\def\Sp{{\rm Sp}}
\def\GL{{\rm GL}}
\def\u{\frak{u}}
\def\Cen{{\rm Cen}}
\def\I{{\rm I}}
\def\T{{\rm T}}
\def\O{{\rm O}}
\def\SO{{\rm SO}}
\def\Aut{{\rm Aut}}
\def\End{{\rm End}}
\def\Ad{{\rm Ad}}
\def\Pin{{\rm Pin}}
\def\Spin{{\rm Spin}}
\newcommand{\mcc}[1]{\overleftarrow{#1}}
\def\Det{{\rm Det}}

\makeindex

\begin{document}

\title[Clifford Algebras and Their Applications to Lie Groups and Spinors]
{CLIFFORD ALGEBRAS AND THEIR APPLICATIONS TO LIE GROUPS AND SPINORS\footnotemark}
\author[Dmitry Shirokov]
{Dmitry Shirokov}
\address{National Research University Higher School of Economics\\
101000 Moscow, Russia; dshirokov@hse.ru             
\\[4pt]
Institute for Information Transmission Problems of Russian Academy of Sciences\\
127051 Moscow, Russia; shirokov@iitp.ru
}

\begin{abstract}
In these lectures, we discuss some well-known facts about Clifford algebras: matrix representations, Cartan's periodicity of 8, double coverings of orthogonal groups by spin groups, Dirac equation in different formalisms, spinors in $n$ dimensions, etc. We also present our point of view on some problems. Namely, we discuss the generalization of the Pauli theorem, the basic ideas of the method of averaging in Clifford algebras, the notion of quaternion type of Clifford algebra elements, the classification of Lie subalgebras of specific type in Clifford algebra, etc.\\[0.2cm]
\end{abstract}

\maketitle
\tableofcontents

\section*{Introduction}\label{sec:1}

Clifford algebra was invented by W.~Clifford \cite{Clifford}. In his research, he combined Hamilton's quaternions \cite{hamilton} and Grassmann's exterior algebra \cite{grassmann}. Further development of the theory of Clifford algebras is associated with a number of famous mathematicians and physicists -- R.~Lipschitz, T.~Vahlen, E.~Cartan \cite{Cartan}, E.~Witt, C.~Chevalley, M.~Riesz \cite{Riesz} and others. Dirac equation \cite{Dirac}, \cite{Dirac2} had a great influence on the development of Clifford algebra. Also note the article \cite{ABS}.

Nowadays Clifford algebra is used in different branches of modern mathematics and physics. There are different applications of Clifford algebra in physics, analysis, geometry, computer science, mechanics, robotics, signal and image processing, etc.

In this text, we discuss some well-known facts about Clifford algebras: matrix representations, Cartan's periodicity of 8, double coverings of orthogonal groups by spin groups, Dirac equation in different formalisms, spinors in $n$ dimensions, etc. We also present our point of view on some problems. Namely, we discuss the generalization of the Pauli theorem, the basic ideas of the method of averaging in Clifford algebras, the notion of quaternion type of Clifford algebra elements, the classification of Lie subalgebras of specific type in Clifford algebra, etc.

We recommend a number of classical books on Clifford algebras and applications \cite{Lounesto}, \cite{Hestenes}, \cite{HestenesSob}, \cite{Chevalley}, \cite{Port}, \cite{Benn:Tucker}, \cite{LM}, \cite{Bud}, \cite{GM}, \cite{DSS}, \cite{GM}, \cite{HM}, \cite{Snygg}, \cite{Snygg2}, \cite{lasenby}, etc.
We can also offer the book \cite{msh3} and a course of lectures \cite{sh8} for Russian speakers.

\smallskip

This text is based on five lectures given by the author at the International Summer School ``Hypercomplex numbers, Lie groups, and Applications'' (Varna, Bulgaria, 9-12 June 2017).

\section{Definition of Clifford Algebra}\label{sec:1}

\subsection{Clifford Algebra as a Quotient Algebra}\label{sec:1.1}

In \cite{Lounesto}, one can find five different (equivalent) definitions of Clifford algebra. We will discuss two definitions of Clifford algebra in this work. Let us start with the definition of Clifford algebra as a quotient algebra \cite{Chevalley}.

\begin{defi}\label{def:1}
Let we have a vector space $V$ of arbitrary finite dimension $n$ over the field $\R$ and a quadratic form $Q: V \to \R$. Consider the tensor algebra
$$T(V)=\bigoplus_{k=0}^\infty \bigotimes^k V$$
and the two-sided ideal $I(V, Q)$ generated by all elements of the form $x \otimes x-Q(x) e$ for $x\in V$, where $e$ is the identity element. Then the following quotient algebra
$$\cl(V,Q)=T(V) \slash I(V,Q)$$
is called real Clifford algebra $\cl(V, Q)$.
\end{defi}

\subsection{Clifford Algebra with Fixed Basis}\label{sec:1.2}

Now let us discuss definition of the real Clifford algebra with fixed basis which is more useful for calculations and different applications.

\begin{defi}\label{def:2}
Let $n$ be a natural number and $E$ be a linear space of dimension $2^n$ over the field of real numbers $\R$ with the basis enumerated by the ordered multi-indices with a length between $0$ and $n$:
$$e, e_{a_1}, e_{a_1 a_2}, \ldots, e_{1\ldots n}$$
where $1\leq a_1 < a_2 < \cdots < a_n \leq n$.
Let us introduce the operation of multiplication on $E$:
\begin{itemize}
  \item with the properties of distributivity, associativity:
\begin{equation}\begin{split}
&U(\alpha V+\beta W)=\alpha UV+\beta UW,\qquad
 U, V, W\in E,\qquad \alpha, \beta \in \R\\
&(\alpha U+\beta V)W=\alpha UW+\beta VW,\qquad
U, V, W\in E,\qquad \alpha, \beta \in \R\\
&U(VW)=(UV)W,\qquad U, V, W\in E
\end{split}\nonumber\end{equation}
  \item $e$ is the identity element:
$$Ue=eU=U,\qquad U\in E$$
  \item $e_a$, $a=1, \ldots, n$ are generators:
$$e_{a_1} e_{a_2} \cdots e_{a_n}=e_{a_1 \ldots a_n},\qquad 1\leq a_1 < a_2 < \cdots < a_n \leq n$$
  \item generators satisfy
$$e_a e_b +e_b e_a =2\eta_{ab}e$$
where
\begin{eqnarray}
\eta=||\eta_{ab}||=\diag(\underbrace{1,\ldots , 1}_p,\underbrace{-1, \ldots, -1}_{q},\underbrace{0, \ldots, 0}_{r}),\qquad p+q+r=n\label{eta}
\end{eqnarray}
is a diagonal matrix with $p$ times $1$, $q$ times $-1$, and $r$ times $0$ on the diagonal.
\end{itemize}
The linear space $E$ with such operation of multiplication is called real Clifford algebra $\cl_{p,q,r}$.
\end{defi}

\begin{ex}
In the case $r=0$, we obtain \emph{nondegenerate Clifford algebra} $\cl_{p,q}:=\cl_{p,q,0}$. The quadratic form $Q$ in Definition \ref{def:1} is nondegenerate in this case.
\end{ex}

\begin{ex}
In the case $r=0$, $q=0$, we obtain \emph{Clifford algebra $\cl_n:=\cl_{n,0,0}$ of Euclidian space}.
The quadratic form $Q$ in Definition \ref{def:1} is positive definite in this case.
\end{ex}

\begin{ex}
In the case $p=q=0$, $r=n$, we obtain \emph{Grassmann algebra} $\Lambda_{n}:=\cl_{0,0,n}$. In this case Clifford multiplication is called \emph{exterior multiplication} and it is denoted by $\wedge$. The generators of Grassmann algebra satisfy conditions $e_a \wedge e_b=-e_b \wedge e_a$, $a, b=1, \ldots, n$.
\end{ex}

Any element of the real Clifford algebra $\cl_{p,q,r}$ has the form
\begin{equation}\label{U}
U=ue+\sum_{a=1}^n u_a e_a + \sum_{a<b} u_{ab} e_{ab}+\cdots+ u_{1\ldots n}e_{1\ldots n}
\end{equation}
where $u, u_a, u_{ab}, \ldots, u_{1\ldots n}\in \R$ are real numbers.

Also we consider \emph{complexified Clifford algebras} $\C\otimes\cl_{p,q,r}$. Any element of the complexified Clifford algebra $\C\otimes\cl_{p,q,r}$ has the form (\ref{U}), where $u$, $u_a$, $u_{ab}$, \ldots, $u_{1\ldots n}\in\C$ are complex numbers.

Also we consider \emph{complex Clifford algebras}. In Definition \ref{def:1}, we must take vector space $V$ over the field of complex numbers $\C$ in this case. In Definition \ref{def:2}, we must take vector space $E$ over the field of complex numbers $\C$ and it is sufficient to consider matrix $\eta=\diag(1,\ldots , 1, 0, \ldots, 0)$, $p+r=n,$ with $p$ times $1$ and $r$ times $0$ on the diagonal instead of the matrix (\ref{eta}) in this case. The most popular case is $\cl({\C}^n)$, when the quadratic form $Q$ is nondegenerate and $\eta$ is the identity matrix.

\subsection{Examples in Small Dimensions}\label{sec:1.3}

\begin{ex}
In the case of $\cl_0$, arbitrary Clifford algebra element has the form $U=ue$, where $e^2=e$. We obtain the isomorphism $\cl_0 \cong \R$.
\end{ex}

\begin{ex}
In the case of $\cl_1$, arbitrary Clifford algebra element has the form $U=ue+u_1 e_1$, where $e$ is the identity element and $e_1^2=e$. We obtain the isomorphism with double numbers: $\cl_1 \cong \R \oplus \R$.
\end{ex}

\begin{ex}
In the case of $\cl_{0,1}$, arbitrary Clifford algebra element has the form $U=ue+u_1 e_1$, where $e$ is the identity element and $e_1^2=-e$. We obtain the isomorphism with complex numbers: $\cl_{0,1}\cong \C$.
\end{ex}

\begin{ex}
In the case of $\cl_{0,2}$, arbitrary Clifford algebra element has the form $U=ue+u_1 e_1+u_2 e_2+u_{12}e_{12}$. We can easily verify the following relations
\begin{eqnarray}
&&(e_1)^2=(e_2)^2=-e\nonumber\\
&&(e_{12})^2=e_1 e_2 e_1 e_2=-e_1 e_1 e_2 e_2=-e\nonumber\\
&&e_1 e_2=-e_2 e_1=e_{12},\qquad e_2 e_{12}=-e_{12} e_2=e_1\nonumber\\
&&e_{12} e_{1}=-e_{1} e_{12}=e_2.\nonumber
\end{eqnarray}
Using the following substitution
$$e_1 \to \ii,\qquad e_2\to \j,\qquad e_{12}\to \k$$
where $\ii$, $\j$, and $\k$ are imaginary units of quaternions, we obtain the isomorphism $\cl_{0,2}\simeq \H$.
\end{ex}

Recall that $\H$ is an associative division algebra. An arbitrary quaternion has the form
$$q=a 1+b\ii+c\j+d\k\in \H,\qquad a, b, c, d\in\R$$
where $1$ is the identity element, $\ii^2=\j^2=\k^2=-1$, $\ii\j=-\j\ii=\k$, $\j\k=-\k\j=\ii$, $\k\ii=-\ii\k=\j$.

If $q\neq 0$, then $q^{-1}=\frac{1}{||q||^2}\bar{q}$, where
$$\bar{q}:=a-b\ii-c\j-d\k,\qquad ||q||:=\sqrt{q \bar{q}}=\sqrt{a^2+b^2+c^2+d^2}.$$

Note that $\cl_{2,0}\cong\cl_{1,1}\ncong\cl_{0,2}$ (see Section \ref{sec:3.1}).

\begin{ex}
Let us consider the Pauli matrices:
$$\sigma_0=\left( \begin{array}{ll}
 1 & 0 \\
 0 & 1 \end{array}\right),\quad
\sigma_1=\left( \begin{array}{ll}
 0 & 1 \\
 1 & 0 \end{array}\right),\quad
\sigma_2=\left( \begin{array}{ll}
 0 & -\ii \\
 \ii & 0 \end{array}\right),\quad
\sigma_3=\left( \begin{array}{ll}
 1 & 0 \\
 0 & -1 \end{array}\right).
$$
W.~Pauli introduced these matrices \cite{Paulimatr} to describe spin of the electron.

It can be easily verified that
\begin{eqnarray}
&&\sigma_1\sigma_2=\ii\sigma_3, \quad\sigma_2\sigma_3=\ii\sigma_1, \quad\sigma_3\sigma_1=\ii\sigma_2\nonumber\\
&&(\sigma_a)^\dagger=\sigma_a,\quad  \tr(\sigma_a)=0,\quad (\sigma_a)^2=\sigma_0,\qquad a=1, 2, 3\nonumber\\
&&\sigma_a \sigma_b=-\sigma_b \sigma_a,\qquad a\neq b,\qquad a,b=1,2,3.\nonumber
\end{eqnarray}
Using the substitution
$$e\to\sigma_{0},\quad e_a\to\sigma_a, a=1, 2, 3,\quad e_{ab}\to \sigma_a \sigma_b, a<b,\quad e_{123}\to \sigma_1\sigma_2\sigma_3$$
we obtain the isomorphism
$$\cl_{3}\cong \Mat(2,\C).$$
The matrices
$$\{\sigma_0,\quad \sigma_1,\quad \sigma_2,\quad \sigma_3,\quad \ii\sigma_1,\quad \ii\sigma_2,\quad \ii\sigma_3,\quad \ii\sigma_0\}$$
constitute a basis of $\Mat(2,\C)$.
\end{ex}

\begin{ex}
Let us consider the Dirac gamma-matrices \cite{Dirac}, \cite{Dirac2}
\begin{eqnarray}
\gamma_0&=&\begin{pmatrix}1 &0 &0 & 0\cr
                  0 &1 & 0&0 \cr
                  0 &0 &-1&0 \cr
                  0 &0 &0 &-1\end{pmatrix},\quad
\gamma_1=\begin{pmatrix}0 &0 &0 & 1\cr
                  0 &0 & 1&0 \cr
                  0 &-1 &0 &0 \cr
                  -1 &0 &0 &0\end{pmatrix}\nonumber
\\
\gamma_2&=&\begin{pmatrix}0 &0 &0 & -\ii\cr
                  0 &0 & \ii&0 \cr
                  0 & \ii&0 &0 \cr
                  -\ii &0 &0 &0\end{pmatrix},\quad
\gamma_3=\begin{pmatrix}0 &0 & 1& 0\cr
                  0 &0 & 0&-1 \cr
                  -1 &0 &0 &0 \cr
                  0 &1&0 &0\end{pmatrix}.
\nonumber
\end{eqnarray}

These matrices satisfy conditions
\begin{eqnarray}
&&\gamma_a\gamma_b+\gamma_b\gamma_a=2\eta_{ab}{\bf1},\quad
a,b=0,1,2,3,\quad \eta=\|\eta_{ab}\|=\diag(1,-1,-1,-1)\nonumber\\
&&\tr \gamma_a=0,\qquad \gamma_a^\dagger=\gamma_0 \gamma_a \gamma_0,\qquad a=0, 1, 2, 3.\nonumber
\end{eqnarray}

Using the substitution $e_{a+1}\to\gamma_a,\quad a=0, 1, 2, 3$, we obtain the isomorphism
$$\C\otimes\cl_{1,3}\cong \Mat(4,\C).$$
\end{ex}

\section{Gradings and Conjugations}\label{sec:2}
\subsection{Gradings}\label{sec:2.1}

Any Clifford algebra element $U\in\cl_{p,q,r}$ has the form
$$U=ue+\sum_a u_a e_a +\sum_{a<b} u_{ab} e_{ab}+\cdots+u_{1\ldots n}e_{1\ldots n}=\sum_A u_A e_A,\quad u_A\in\R$$
where we denote arbitrary ordered multi-index by $A=a_1\ldots a_k$. Denote its length by $|A|=k$.

\begin{defi}
The following subspace
$$\cl^k_{p,q,r}=\{\sum_{|A|=k} u_A e_A\}$$
is called subspace of grade $k$.
\end{defi}

We have
$$\cl_{p,q,r}=\bigoplus_{k=0}^{n}\cl^k_{p,q,r},\qquad \dim \cl^k_{p,q,r}=C^k_n=\frac{n!}{k! (n-k)!}.$$

Let us consider \emph{projection operations onto subspaces of grade $k$}:
$$U\in\cl_{p,q,r}\to <\!\!U\!\!>_k \in \cl_{p,q,r}^k.$$
Note that for arbitrary element $U\in\cl_{p,q,r}$ we have
$$U=\sum_{k=0}^n <\!\!U\!\!>_k.$$

The Clifford algebra $\cl_{p,q,r}$ is \emph{a $Z_2$-graded algebra}. It can be represented in the form of the direct sum of \emph{even and odd subspaces}:
$$
\cl_{p,q,r}=\cl^{(0)}_{p,q,r}\oplus\cl^{(1)}_{p,q,r}$$
where
$$
\cl^{(0)}_{p,q,r}=\!\!\!\bigoplus_{k=0{\rm mod} 2}\!\!\!\cl^k_{p,q,r},\quad \cl^{(1)}_{p,q,r}=\!\!\!\bigoplus_{k=1 {\rm mod} 2}\!\!\!\cl^k_{p,q,r}.
$$
We have
$$
\cl^{(i)}_{p,q,r} \cl^{(j)}_{p,q,r}\subset \cl^{(i+j){\rm mod} 2}_{p,q,r},\quad i=0, 1
$$
and
$$
\dim \cl^{(0)}_{p,q,r}=\dim \cl^{(1)}_{p,q,r}=2^{n-1}.
$$
Note that $\cl^{(0)}_{p,q,r}$ is a subalgebra of $\cl_{p,q,r}$.

\subsection{Center of Clifford Algebra}\label{sec:2.2}

We have the following well-known theorem about \emph{the center of Clifford algebra}
$\Cen(\cl_{p,q}):=\{U\in\cl_{p,q} \semicolon UV=VU \,\mbox{for all}\, V\in\cl_{p,q}\}$.

\begin{theorem}\label{thcenter} The center of the Clifford algebra $\cl_{p,q}$ is
$$\Cen(\cl_{p,q})=\left\lbrace
\begin{array}{ll}
\cl^0_{p,q}=\{ue \semicolon u\in\R\}, & \parbox{.5\linewidth}{if $n$ is even}\\
\cl^0_{p,q}\oplus\cl^n_{p,q}=\{ue+u_{1\ldots n}e_{1\ldots n} \semicolon u, u_{1\ldots n}\in\R\}, & \parbox{.5\linewidth}{if $n$ is odd.}
\end{array}\nonumber
\right.$$
\end{theorem}

\begin{proof}
Let us represent element $U$ in the form
$$U=U^{(0)}+U^{(1)},\qquad U^{(i)}\in\cl^{(i)}_{p,q}, i=0, 1.
$$
We have conditions $UV=VU$ for all $V\in\cl_{p,q}$. We obtain
$$U^{(i)} e_k=e_k U^{(i)},\qquad k=1,\ldots, n,\qquad i=0, 1.$$

We represent $U^{(0)}$ in the form $U^{(0)}=A^{(0)}+e_1 B^{(1)}$, where $A^{(0)}\in\cl^{(0)}_{p,q}$ and $B^{(1)}\in\cl^{(1)}_{p,q}$ do not contain $e_1$. For $k=1$ we obtain
$$(A^{(0)}+e_1 B^{(1)})e_1=e_1 (A^{(0)}+e_1 B^{(1)}).$$
Using $A^{(0)}e_1=e_1 A^{(0)}$ and $e_1 B^{(1)}e_1=-e_1 e_1 B^{(1)}$, we obtain $B^{(1)}=0$. Acting similarly for $e_2, \ldots, e_n$, we obtain $U^{(0)}=ue$.

We represent $U^{(1)}$ in the form $U^{(1)}=A^{(1)}+e_1 B^{(0)}$, where $A^{(1)}\in\cl^{(1)}_{p,q}$ and $B^{(0)}\in\cl^{(0)}_{p,q}$ do not contain $e_1$. For $k=1$ we obtain
$$(A^{(1)}+e_1 B^{(0)})e_1=e_1 (A^{(1)}+e_1 B^{(0)}).$$
Using $A^{(1)}e_1=-e_1 A^{(1)}$ and $e_1 B^{(0)}e_1=e_1 e_1 B^{(0)}$, we obtain $A^{(1)}=0$. Acting similarly for $e_2, \ldots, e_k$, we obtain $U^{(1)}=u_{1\ldots n}e_{1\ldots n}$ in the case of odd $n$ and $U^{(1)}=0$ in the case of even $n$.
\end{proof}

\subsection{Operations of Conjugations}\label{sec:2.3}

\begin{defi} The following operation (involution) in the Clifford algebra $\cl_{p,q,r}$
$$\widehat{U}:=U|_{e_a\to -e_a},\qquad U\in\cl_{p,q,r}$$
is called grade involution or main involution.
\end{defi}

It can be verified that
$$
\widehat{U}=\sum_{k=0}^n <\!\!\widehat{U}\!\!>_k=\sum_{k=0}^n (-1)^k <\!\!U\!\!>_k.
$$
We have
$$
\widehat{\widehat{U}}=U,\quad \widehat{UV}=\widehat{U} \widehat{V},\quad \widehat{\lambda U+\mu V}=\lambda \widehat{U}+\mu\widehat{V},\quad U, V\in\cl_{p,q}, \lambda, \mu\in\R.
$$

\begin{defi} The following operation (anti-involution) in the Clifford algebra $\cl_{p,q,r}$ $$\widetilde{U}:=U|_{e_{a_1\ldots a_k}\to e_{a_k}\ldots e_{a_1}},\qquad U\in\cl_{p,q,r}$$
is called reversion.
\end{defi}

We have
$$\widetilde{U}=\sum_{k=0}^n <\!\!\widetilde{U}\!\!>_k=\sum_{k=0}^n (-1)^{\frac{k(k-1)}{2}} <\!\!U\!\!>_k$$
and
$$
\widetilde{\widetilde{U}}=U,\quad \widetilde{UV}=\widetilde{V}\widetilde{U},\quad \widetilde{\lambda U+\mu V}=\lambda \widetilde{U}+\mu \widetilde{V},\quad U, V\in\cl_{p,q}, \lambda, \mu\in \R.
$$

\begin{defi} A superposition of reversion and grade involution is called Clifford conjugation.
\end{defi}

We do not use individual notation for Clifford conjugation and use notation $\widehat{\widetilde{U}}$. The operation of Clifford conjugation corresponds to the operation of complex conjugation of complex numbers in the case $\cl_{0,1}\cong\C$ and quaternion conjugation in the case $\cl_{0,2}\cong\H$.

We have
$$\widehat{\widetilde{U}}=\sum_{k=0}^n <\!\!\widehat{\widetilde{U}}\!\!>_k=\sum_{k=0}^n (-1)^{\frac{k(k+1)}{2}} <\!\!U\!\!>_k.$$

\begin{defi} The following operation in the complexified Clifford algebra $\C\otimes\cl_{p,q}$
$$
\overline{U}:=U|_{u_{a_1\ldots a_k}\to \bar{u}_{a_1 \ldots a_k}},\qquad U\in\C\otimes\cl_{p,q,r}
$$
where we take complex conjugation of complex numbers $u_{a_1\ldots a_k}$, is called complex conjugation in Clifford algebra.
\end{defi}

We have
$$
\overline{\overline{U}}=U,\quad \overline{UV}=\overline{U}\overline{V},\quad \overline{\lambda U+\mu V}=\bar{\lambda} \overline{U}+\overline{\mu}\overline{V},\quad  U, V\in\C\otimes\cl_{p,q},  \lambda, \mu \in \C.
$$

An important operation of Hermitian conjugation in $\C\otimes\cl_{p,q}$ will be considered in Section \ref{sec:3.4}.

\subsection{Quaternion Types of Clifford Algebra Elements}\label{sec:2.4}

The operation of grade involution uniquely determines two (even and odd) subspaces of the Clifford algebra:
$$
\cl^{(j)}_{p,q,r}:=\bigoplus_{k=j {\rm mod}2}\!\!\!\cl^k_{p,q,r}=\{U\in\cl_{p,q,r} \semicolon \widehat{U}=(-1)^j U\},\,\, j=0, 1.
$$

In a similar way, operations of grade involution and reversion uniquely determine the following four subspaces. This is symbolically shown in Table \ref{tab:1}. Instead of question marks, depending on the case, the signs ``plus'' or ``minus'' should be used.

\begin{defi} The following four subspaces of $\cl_{p,q,r}$
$$
\overline{\textbf{j}}:=\bigoplus_{k=j {\rm mod}4}\cl^k_{p,q,r}=\{U\in\cl_{p,q,r} \semicolon \widehat{U}=(-1)^j U, \widetilde{U}=(-1)^{\frac{j(j-1)}{2}} U\}
$$
are called subspaces of quaternion types $j=0, 1, 2, 3$.
\end{defi}

\begin{table}[ht]
\centering
\begin{tabular}{|l|l|l|l|l|}
\hline
$\cl_{p,q,r}$ & $\overline{\textbf{0}}$ & $\overline{\textbf{1}}$ & $\overline{\textbf{2}}$ & $\overline{\textbf{3}}$ \\ \hline
  $\widehat{U}= ? U$ & + & - & + & - \\ \hline
  $\widetilde{U}= ? U$ & + & + & - & - \\
\hline
\end{tabular}
\medskip
\caption{Subspaces of quaternion types in $\cl_{p,q,r}$}\label{tab:1}
\end{table}

We have
$$
\cl_{p,q,r}=\overline{\textbf{0}}\oplus\overline{\textbf{1}}\oplus\overline{\textbf{2}} \oplus\overline{\textbf{3}}, \qquad \cl^{(0)}_{p,q,r}=\overline{\textbf{0}}\oplus\overline{\textbf{2}},\qquad
\cl^{(1)}_{p,q,r}=\overline{\textbf{1}}\oplus\overline{\textbf{3}}.$$

Grade involution, reversion, and complex conjugation uniquely determine eight subspaces of the complexified Clifford algebra. This is symbolically shown in Table \ref{tab:2}. Instead of question marks, depending on the case, the signs ``plus'' or ``minus'' should be used.

We have
$$
\C\otimes\cl_{p,q,r}=\overline{\textbf{0}}\oplus\overline{\textbf{1}}\oplus\overline{\textbf{2}} \oplus\overline{\textbf{3}}\oplus \ii \overline{\textbf{0}}\oplus \ii\overline{\textbf{1}}\oplus \ii\overline{\textbf{2}}\oplus \ii\overline{\textbf{3}}.
$$

\begin{table}[ht]
\centering
\begin{tabular}{|l|l|l|l|l|l|l|l|l|}
\hline
  $\C\otimes\cl_{p,q,r}$ & $\overline{\textbf{0}}$ & $\overline{\textbf{1}}$ & $\overline{\textbf{2}}$ & $\overline{\textbf{3}}$ & $\ii\overline{\textbf{0}}$ & $\ii\overline{\textbf{1}}$ & $\ii\overline{\textbf{2}}$ & $\ii\overline{\textbf{3}}$ \\ \hline
  $\widehat{U}= ? U$ & + & - & + & - & + & - & + & - \\ \hline
  $\widetilde{U}= ? U$ & + & + & - & - & + & + & - & -\\ \hline
  $\bar{U}= ? U$ & + & + & + &+ & - & - & - & -\\
\hline
\end{tabular}
\medskip
\caption{Subspaces of quaternion types in $\C\otimes\cl_{p,q,r}$}\label{tab:2}
\end{table}

The subspaces of quaternion types have the following dimensions
\begin{eqnarray}
&&\dim \overline{\textbf{0}}=\sum_k C_n^{4k}=2^{n-2}+2^{\frac{n-2}{2}}\cos{\frac{\pi n}{4}}\nonumber\\
&&\dim \overline{\textbf{1}}=\sum_k C_n^{4k+1}=2^{n-2}+2^{\frac{n-2}{2}}\sin{\frac{\pi n}{4}}\label{dimquat}\\
&&\dim \overline{\textbf{2}}=\sum_k C_n^{4k+2}=2^{n-2}-2^{\frac{n-2}{2}}\cos{\frac{\pi n}{4}}\nonumber\\
&&\dim \overline{\textbf{3}}=\sum_k C_n^{4k+3}=2^{n-2}-2^{\frac{n-2}{2}}\sin{\frac{\pi n}{4}}.\nonumber
\end{eqnarray}

We denote the commutator of two Clifford algebra elements $U, V$ by $[U,V]:=UV-VU$ and the anticommutator by $\{U, V\}:=UV+VU$.

\begin{theorem}\label{qt}\cite{sh1}, \cite{sh9} We have the following properties:
\begin{eqnarray}
&&[\overline{\textbf{j}}, \overline{\textbf{j}}]\subset \overline{\textbf{2}},\quad [\overline{\textbf{j}}, \overline{\textbf{2}}]\subset \overline{\textbf{j}},\quad j=0, 1, 2, 3\nonumber\\
&&[\overline{\textbf{0}}, \overline{\textbf{1}}]\subset \overline{\textbf{3}},\quad [\overline{\textbf{0}}, \overline{\textbf{3}}]\subset \overline{\textbf{1}},\quad [\overline{\textbf{1}}, \overline{\textbf{3}}]\subset \overline{\textbf{0}}\nonumber\\
&&\{\overline{\textbf{j}}, \overline{\textbf{j}}\}\subset \overline{\textbf{0}},\quad \{\overline{\textbf{j}}, \overline{\textbf{0}}\}\subset \overline{\textbf{j}},\quad j=0, 1, 2, 3\nonumber\\
&&\{\overline{\textbf{1}}, \overline{\textbf{2}}\}\subset \overline{\textbf{3}},\quad \{\overline{\textbf{2}}, \overline{\textbf{3}}\}\subset \overline{\textbf{1}},\quad \{\overline{\textbf{3}}, \overline{\textbf{1}}\}\subset \overline{\textbf{2}}.\nonumber
\end{eqnarray}
\end{theorem}

By Theorem \ref{qt}, the Clifford algebra $\cl_{p,q,r}$ is a $Z_2\times Z_2$-graded algebra w.r.t. the operation of commutator and w.r.t. the operation of anticommutator.

The notion of quaternion type was introduced by the author in the brief report \cite{sh1} and the paper \cite{sh9}. Further development of this concept is given in \cite{sh5}, \cite{sh10}, \cite{sh14}, \cite{sh19}, see also books \cite{msh3}, \cite{sh8}.

Subspaces of quaternion types are useful in different calculations (see \cite{sh5}, \cite{sh10}, \cite{sh14}). Here and below we omit the sign of the direct sum to simplify notation: $\overline{\textbf{0}}\oplus\overline{\textbf{1}}=\overline{\textbf{01}}$, $\overline{\textbf{0}}\oplus\overline{\textbf{1}}\oplus\overline{\textbf{2}}\oplus\overline{\textbf{3}}= \overline{\textbf{0123}}=\cl_{p,q}$, etc.

For example, if $U\in \overline{\textbf{k}}$ for some $k=0, 1, 2, 3$, then (see \cite{sh10})
$$U^m \in \left\lbrace
\begin{array}{ll}
\overline{\textbf{k}}, & \mbox{ if $m$ is odd}\\
\overline{\textbf{0}}, & \mbox{ if $m$ is even,}
\end{array}\nonumber
\right.\qquad \sin U \in \overline{\textbf{k}},\qquad \cos U \in \overline{\textbf{0}}.
$$
For arbitrary element $U\in\cl_{p,q}$ we have (see \cite{sh5})
$$ U\widetilde{U},\, \widetilde{U}U \in \overline{\textbf{01}},\qquad U\widehat{\widetilde{U}},\, \widehat{\widetilde{U}}U \in  \overline{\textbf{03}}.$$

Using the classification of Clifford algebra elements based on the notion of quaternion type, we present a number of Lie algebras in $\C\otimes\cl_{p,q}$ (see Section \ref{sec:6.4} and \cite{sh19}).

\section{Matrix Representations of Clifford Algebras}\label{sec:3}

\subsection{Cartan's Periodicity of 8; Central and Simple Algebras}\label{sec:3.1}

\begin{lemma}\cite{Lounesto}\label{lemma} We have the following isomorphisms of associative algebras:
\begin{eqnarray}
&&1) \cl_{p+1, q+1}\cong \Mat(2, \cl_{p,q}),\qquad 2) \cl_{p+1, q+1}\cong \cl_{p,q}\otimes \cl_{1,1}\nonumber\\
&&3) \cl_{p, q}\cong \cl_{q+1, p-1},\quad p\geq 1,\qquad 4) \cl_{p,q}\cong\cl_{p-4, q+4},\quad p\geq 4.\nonumber
\end{eqnarray}
\end{lemma}

\begin{proof} Let $e_1, \ldots, e_n$ be the generators of $\cl_{p,q}$ and $(e_+)^2=e$, $(e_-)^2=-e$ such that all generators $e_1, \ldots, e_n, e_+, e_-$ anticommute with each other.
\begin{enumerate}
  \item We obtain generators of $\Mat(2, \cl_{p,q})$ in the following way:
$$e_i \to \left(
            \begin{array}{cc}
              e_i & 0 \\
              0 & -e_i \\
            \end{array}
          \right),\quad i=1, \ldots, n,\quad
e_+\to \left(
         \begin{array}{cc}
           0 & e \\
           e & 0 \\
         \end{array}
       \right),\quad
e_- \to \left(
          \begin{array}{cc}
            0 & -e \\
            e & 0 \\
          \end{array}
        \right).$$
  \item $e_i e_+ e_-$, $i=1, \ldots, n$ are generators of $\cl_{p,q}$ and $e_+$, $e_-$ are generators of $\cl_{1,1}$. Each generator of $\cl_{p,q}$ commutes with each generator of $\cl_{1,1}$.
  \item $e_1$, $e_i e_1$, $i=2, \ldots, n$ are generators of $\cl_{q+1, p-1}$.
  \item $e_i e_1 e_2 e_3 e_4$, $i=1, 2, 3, 4$ and $e_j$, $j=5, \ldots, n$ are generators of $\cl_{p-4, q+4}$.
\end{enumerate}
\end{proof}

We have the following well-known theorems about isomorphisms between Clifford algebras and matrix algebras.

\begin{theorem}[Cartan 1908]\label{Cartan} We have the following isomorphism of algebras
\begin{eqnarray}
\cl_{p,q}\cong\left\lbrace
\begin{array}{ll}
\Mat(2^{\frac{n}{2}},\R), & \parbox{.5\linewidth}{ if $p-q\equiv0; 2\!\!\mod 8$}\\
\Mat(2^{\frac{n-1}{2}},\R)\oplus \Mat(2^{\frac{n-1}{2}},\R), & \parbox{.5\linewidth}{ if $p-q\equiv1\!\!\mod 8$}\\
\Mat(2^{\frac{n-1}{2}},\C), & \parbox{.5\linewidth}{ if $p-q\equiv3; 7\!\!\mod 8$}\\
\Mat(2^{\frac{n-2}{2}},\H), & \parbox{.5\linewidth}{ if $p-q\equiv4; 6\!\!\mod 8$}\\
\Mat(2^{\frac{n-3}{2}},\H)\oplus \Mat(2^{\frac{n-3}{2}},\H), & \parbox{.5\linewidth}{ if $p-q\equiv5\!\!\mod 8$.}
\end{array}\nonumber
\right.
\end{eqnarray}
\end{theorem}

\begin{proof} Using Lemma \ref{lemma}, we obtain isomorphisms for all $\cl_{p,q}$ (see Table \ref{tab:3}).
We use notations $^2\R:=\R\oplus\R,\quad \R(2):=\Mat(2,\R),\ldots$

We know isomorphisms (see Section \ref{sec:1.3})
$$\cl_{0,0}\cong \R,\quad \cl_{0,1}\cong\C,\quad \cl_{1,0}\cong\R\oplus\R,\quad \cl_{0,2}\cong\H.$$

Using the substitution
$$e \to (1,1),\quad e_1 \to (\ii, -\ii),\quad e_2 \to (\j, -\j),\quad e_3 \to (\k, -\k)$$
we obtain the isomorphism
$$\cl_{0,3}\cong \H\oplus\H.$$

Using $\cl_{p+1, q+1}\cong \Mat(2, \cl_{p,q})$, we get $\cl_{1,1}\cong \Mat(2,\R)$. Using $\cl_{p+1, q+1}\cong \cl_{p,q}\otimes \cl_{1,1}$, we conclude that if we make a step down Table \ref{tab:3} ($n\to n+2$), then the size of corresponding matrix algebra is doubled ($\Mat(k, \ldots) \to \Mat(2k, \ldots)$). Using $\cl_{p, q}\cong \cl_{q+1, p-1}$, we conclude that Table \ref{tab:3} is symmetric w.r.t. the column ``$p-q=1$''. Using $\cl_{p,q}\cong\cl_{p-4, q+4}$, we obtain the symmetry $p-q \leftrightarrow p-q-8$ for each $n$.
\end{proof}

\begin{table}[ht]
\centering
\begin{tabular}{|p{0.09\linewidth}|p{0.045\linewidth}p{0.045\linewidth} p{0.045\linewidth}p{0.045\linewidth}p{0.045\linewidth}p{0.045\linewidth}p{0.045\linewidth} p{0.045\linewidth}p{0.045\linewidth}p{0.045\linewidth}p{0.07\linewidth}|}
\hline
$n \backslash p-q$ &  $-5$ & $-4$ & $-3$ & $-2$ & $-1$ & $0$ & $1$ & $2$ & $3$ & $4$ & $5$ \\ \hline
$0$ &  $-$ & $-$ & $-$ & $-$ & $-$ & $\R$ & $-$ & $-$ & $-$ & $-$ & $-$  \\
$1$ &  $-$ & $-$ & $-$ & $-$ & $\C$ & $-$ & $^2\R$ & $-$ & $-$ & $-$ & $-$ \\
$2$ &  $-$ & $-$ & $-$ & $\H$ & $-$ & $\R(2)$ & $-$ & $\R(2)$ & $-$ & $-$ & $-$ \\
$3$ &  $-$ & $-$ & $^2\H$ & $-$ & $\C(2)$ & $-$ & $^2\R(2)$ & $-$ & $\C(2)$ & $-$ & $-$ \\
$4$ &  $-$ & $\H(2)$ & $-$ & $\H(2)$ & $-$ & $\R(4)$ & $-$ & $\R(4)$ & $-$ & $\H(2)$ & $-$ \\
$5$ &  $\C(4)$ & $-$ & $^2\H(2)$ & $-$ & $\C(4)$ & $-$ & $^2\R(4)$ & $-$ & $\C(4)$ & $-$ & $^2\H(2)$ \\
\hline
\end{tabular}
\medskip
\caption{Isomorhisms between $\cl_{p,q}$ and matrix algebras}\label{tab:3}
\end{table}

Similarly, we can obtain the following isomorphisms for complex Clifford algebras and for even subalgebras of the Clifford algebra.

\begin{theorem}\cite{Lounesto} We have the following isomorphism of algebras
\begin{eqnarray}
\cl(\C^n)\cong\C\otimes\cl_{p,q}\cong\left\lbrace
\begin{array}{ll}
\Mat(2^{\frac{n}{2}}, \C), & \parbox{.5\linewidth}{if $n$ is even}\\
\Mat(2^{\frac{n-1}{2}}, \C)\oplus \Mat(2^{\frac{n-1}{2}}, \C), & \parbox{.5\linewidth}{if $n$ is odd.}
\end{array}\nonumber
\right.
\end{eqnarray}
\end{theorem}

\begin{theorem}\cite{Lounesto}\label{thEven} We have the following isomorphism of algebras
$$1) \cl^{(0)}_{p,q}\cong \cl_{p, q-1},\quad q\geq 1;\qquad 2) \cl^{(0)}_{p,q}\cong \cl_{q, p-1},\quad p\geq 1;\qquad 3)\cl^{(0)}_{p,q}\cong \cl^{(0)}_{q,p}.$$
\end{theorem}

\begin{proof}
 Let $e_1, \ldots, e_{n}$ be the generators of $\cl_{p,q}$.
 \begin{enumerate}
   \item Then $e_i e_n$, $i=1, \ldots, n-1$ are generators of $\cl^{(0)}_{p, q}$.
   \item Then $\left\{
                 \begin{array}{ll}
                   e_{p+i}e_p, & \hbox{$i=1, \ldots, q$} \\
                   e_{j-q} e_p, & \hbox{$j=q+1, \ldots, n-1$}
                 \end{array}
               \right.$
are generators of $\cl^{(0)}_{p,q}$.
   \item Using 1) and 2), we get 3).
 \end{enumerate}
 \end{proof}

\begin{defi}
An algebra is simple if it contains no non-trivial two-sided ideals and the multiplication operation is not zero.
\end{defi}

\begin{defi}
A central simple algebra over a field $\FF$ is a finite-dimensional associative algebra, which is simple, and for which the center is exactly $\FF$.
\end{defi}

The following classification of Clifford algebras can be found in \cite{Chevalley}.
\begin{itemize}
  \item If $n$ is even, then $\cl(V,Q)$ is a central simple algebra.
  \item If $n$ is odd and $\FF=\C$, then $\cl(V,Q)$ is the direct sum of two isomorphic complex central simple algebras.
  \item If $n$ is odd, $\FF=\R$, and $(e_{1\ldots n})^2=e$, then $\cl(V,Q)$ is the direct sum of two isomorphic simple algebras.
  \item If $n$ is odd, $\FF=\R$, and $(e_{1\ldots n})^2=-e$, then $\cl(V,Q)$ is simple with center $\cong\C$.
\end{itemize}

Note that
$$
(e_{1\ldots n})^2=(-1)^{q+\frac{n(n-1)}{2}}e=\left\{
                                                 \begin{array}{ll}
                                                   e, & \hbox{if $p-q=0, 1\mod 4$} \\
                                                   -e, & \hbox{if $p-q=2, 3\mod 4$}
                                                 \end{array}
                                               \right.
$$
and these results agree with Theorem \ref{Cartan}.

\subsection{Clifford Trigonometry Circle and Exterior Signature of Clifford Algebra}\label{sec:3.2}

In the literature, the Cartan's periodicity of 8 is depicted in the form of a eight-hour clock (see Table \ref{tab:4}). The clockwise movement by one step corresponds to an increase of $p-q$ by 1. The clock shows that two Clifford algebra with the same $p-q\mod 8$ are Morita equivalent (see Theorem \ref{Cartan}).

\begin{table}[ht]
\centering
\begin{tabular}{|lllllllll|} \hline
& &   & & 0    &  &  & &\\
&7 &   & & $\R$ &  &  &  1 &\\
& & $^2\R$ & &  &  & $\C$ & &  \\
& &   & &  &  &  &  &\\
6&$\R$  &   & &  &  &  &$\H$ & 2  \\
 &  &   & &  &  &  & & \\
  &  & $\C$  & &  &  & $^2\H$  &  &\\
   & 5 &   & & $\H$ &  &  & 3 &\\
   &  &   & &  4   &  &  &  & \\ \hline
\end{tabular}
\medskip
\caption{Clifford clock or Clifford trigonometry circle}\label{tab:4}
\end{table}

In our opinion, it is more correct to call it not Clifford clock, but \emph{Clifford trigonometry circle}.
The algebras on the clock are in one-to-one correspondence with the values of the function $\sin\frac{\pi(p-q+1)}{4}$. To show this, we do the following calculations (see also \cite{msh3}).

Let $P$ be the number of basis elements $e_A$ of the Clifford algebra $\cl_{p,q}$ such that $(e_A)^2=e$ and $Q$ be the number of basis elements $e_A$ of $\cl_{p,q}$ such that $(e_A)^2=-e$. We call $(P,Q)$ \emph{exterior signature of Clifford algebra $\cl_{p,q}$}, $P+Q=2^n$.

\begin{theorem} We have
$$P=2^{\frac{n-1}{2}}(2^{\frac{n-1}{2}}+\sin\frac{\pi(p-q+1)}{4}),\qquad Q=2^{\frac{n-1}{2}}(2^{\frac{n-1}{2}}-\sin\frac{\pi(p-q+1)}{4}).$$
\end{theorem}

Note that
$$\sum_{A}(e_A)^2=P-Q=2^{\frac{n+1}{2}}\sin\frac{\pi(p-q+1)}{4}.$$

Finally, we obtain the following theorem.
\begin{theorem} Two Clifford algebras $\cl_{p_1, q_1}$ and $\cl_{p_2, q_2}$ are isomorphic if and only if their exterior signatures coincide $(P_1, Q_1)=(P_2, Q_2)$ (or, equivalently, $P_1-Q_1=P_2-Q_2$).
\end{theorem}

\subsection{Trace, Determinant and Inverse of Clifford Algebra Elements}\label{sec:3.3}

\begin{defi} The following projection operation onto subspace $\C\otimes\cl^0_{p,q}$
$$\Tr(U):=<\!U\!>|_{e \to 1}$$
is called a trace of Clifford algebra element.
\end{defi}

We have
$$\Tr(U)=u,\qquad U=ue+\sum_a u_a e_a+\cdots+u_{1\ldots n}e_{1\ldots n}.$$

\begin{theorem} \cite{sh7}, \cite{msh3} The operation of trace of Clifford algebra elements has the following properties
\begin{eqnarray}
&&\Tr(U+V)=\Tr(U)+\Tr(V),\quad \Tr(\lambda U)=\lambda \Tr(U),\quad \Tr(UV)=\Tr(VU)\nonumber\\
&&\Tr(UVW)=\Tr(VWU)=\Tr(WUV),\quad U,V,W\in\C\otimes\cl_{p,q}, \lambda\in\C\nonumber\\
&&\Tr(U^{-1}VU)=\Tr(V),\quad \Tr(U)=\Tr(\hat{U})=\Tr(\tilde{U})=\overline{\Tr{\bar{U}}}.\nonumber
\end{eqnarray}
\end{theorem}

We have the following relation with the trace of matrices.
\begin{theorem}\cite{sh7}, \cite{msh3} We have
$$\Tr(U)=\frac{1}{2^{[\frac{n+1}{2}]}}\tr(\gamma(U))$$
where
$$\gamma: \C\otimes\cl_{p,q}\to\left\{
                     \begin{array}{ll}
                       \Mat(2^{\frac{n}{2}},\C), & \hbox{if $n$ is even} \\
                       \Mat(2^{\frac{n-1}{2}},\C)\oplus\Mat(2^{\frac{n-1}{2}},\C), & \hbox{if $n$ is odd}
                     \end{array}
                   \right.$$
is faithful matrix representation of $\C\otimes\cl_{p,q}$ (of minimal dimension).
\end{theorem}

\begin{defi}\label{defidet} The determinant of any faithful matrix representation (of minimal dimension) of element $U$ is called the determinant of Clifford algebra element $U\in\C\otimes\cl_{p,q}$.
\end{defi}

\begin{theorem}\cite{sh7}, \cite{msh3}
Definition \ref{defidet} is correct: the determinant does not depend on the choice of matrix representation.
\end{theorem}

\begin{theorem}\cite{sh7}, \cite{msh3} The determinant of Clifford algebra elements has the following properties:
\begin{itemize}
  \item We have
\begin{eqnarray}
&&\Det(UV)=\Det(U)\Det(V),\quad \Det(\lambda U)=\lambda^{2^{[\frac{n+1}{2}]}} \Det(U)\nonumber\\
&&\Det(U)=\Det(\widehat{U})=\Det(\widetilde{U})=\overline{\Det(\overline{U})},\quad U, V\in\C\otimes\cl_{p,q}, \lambda\in\C.\nonumber
\end{eqnarray}
  \item $U^{-1}\in \C\otimes\cl_{p,q}$ exists if and only if $\Det U\neq 0$.
  \item If $U^{-1}$ exists, then
  $$\Det (U^{-1}) =(\Det U)^{-1},\, \Det(U^{-1}VU)=\Det(V),\,  V\in\C\otimes\cl_{p,q}.$$
\end{itemize}
\end{theorem}

\begin{theorem} \cite{sh7}, \cite{msh3} We have the following explicit formulas for the determinant and the inverse of Clifford algebra element $U\in\C\otimes\cl_{p,q}$ in the cases $n=1, \ldots, 5$:
\begin{equation}
\Det\,U=\left\lbrace
\begin{array}{ll}
U|_{e\to 1}, & n=0\\
U \widehat{U}|_{e\to 1}, & n=1\\
U \widehat{\widetilde{U}}|_{e\to 1}, & n=2\\
U \widetilde{U} \widehat{U} \widehat{\widetilde{U}}|_{e\to 1}= U \widehat{\widetilde{U}} \widehat{U} \widetilde{U}|_{e\to 1}, & n=3\\
U \widetilde{U} (\widehat{U} \widehat{\widetilde{U}})^\bigtriangledown|_{e\to 1}= U \widehat{\widetilde{U}} (\widehat{U} \widetilde{U})^\bigtriangledown|_{e\to 1}, & n=4\\
U \widetilde{U} (\widehat{U} \widehat{\widetilde{U}})^\bigtriangledown (U \widetilde{U} (\widehat{U} \widehat{\widetilde{U}})^\bigtriangledown)^\bigtriangleup|_{e\to 1}, & n=5
\end{array}
\right.\nonumber
\end{equation}
\begin{equation}
(U)^{-1}=\frac{1}{\Det\,U}\left\lbrace
\begin{array}{ll}
e, & n=0\\
\widehat{U}, & n=1\\
\widehat{\widetilde{U}}, & n=2\\
\widetilde{U} \widehat{U} \widehat{\widetilde{U}} \quad(\mbox{or}\quad \widehat{\widetilde{U}} \widehat{U} \widetilde{U}), & n=3\\
\widetilde{U} (\widehat{U} \widehat{\widetilde{U}})^\bigtriangledown \quad(\mbox{or}\quad \widehat{\widetilde{U}}(\widehat{U} \widetilde{U})^\bigtriangledown), & n=4\\
\widetilde{U} (\widehat{U} \widehat{\widetilde{U}})^\bigtriangledown (U \widetilde{U} (\widehat{U} \widehat{\widetilde{U}})^\bigtriangledown)^\bigtriangleup, & n=5
\end{array}
\right.\nonumber
\end{equation}
where $U^\bigtriangledown=U|_{<U>_4\to -<U>_4, <U>_5 \to -<U>_5}$ and $U^\bigtriangleup=U|_{<U>_5 \to -<U>_5}$.
\end{theorem}

Note that we can introduce the notions of trace and determinant of elements of the real Clifford algebra $\cl_{p,q}$. These operations have similar properties (see \cite{sh7}, \cite{msh3}).

\subsection{Unitary Space on Clifford Algebra}\label{sec:3.4}

\begin{theorem}\cite{msh1}, \cite{msh3} The operation $U, V\in\C\otimes\cl_n \to(U,V):=\Tr(\bar{\widetilde{U}}V)$ is a Hermitian (or Euclidian) scalar product on $\C\otimes\cl_{n}$ (or $\cl_n$ respectively).
\end{theorem}

\begin{proof}
We must verify
\begin{eqnarray}
&&(U,V)=\overline{(V,U)},\quad (U,\lambda V)=\lambda (U,V),\quad (U,V+W)=(U,V)+(U,W)\nonumber\\
&&(U,U)\geq 0,\qquad (U,U)=0 \Leftrightarrow U=0\label{uu}
\end{eqnarray}
for all $U, V, W\in\C\otimes\cl_{p,q}$, $\lambda\in \C$. To prove (\ref{uu}) it is sufficient to prove that the basis of $\C\otimes\cl_{p,q}$ is orthonormal:
$$(e_{i_1\ldots i_k}, e_{j_1 \ldots j_l})=\Tr(e_{i_k}\cdots e_{i_1}e_{j_1}\cdots e_{j_l})=
\left\{
                                            \begin{array}{ll}
                                              1, & \hbox{if $(i_1, \ldots, i_k)=(j_1, \ldots, j_l)$} \\
                                              0, & \hbox{if $(i_1, \ldots, i_k)\neq(j_1, \ldots, j_l)$.}
                                            \end{array}
                                          \right.
$$
We have $$(U,U)=\sum_A u_A \overline{u_A}=\sum_A |u_A|^2\geq 0.$$
The theorem is proved.
\end{proof}

\begin{defi}\label{defHerm} Let us consider the following operation of Hermitian conjugation in Clifford algebra:
\begin{eqnarray}U^\dagger:=U|_{e_{a_1\ldots a_k} \to e_{a_1\ldots a_k}^{-1},\, u_{a_1\ldots a_k}\to \bar{u}_{a_1 \ldots a_k}},\qquad U\in\C\otimes\cl_{p,q}.\nonumber
\end{eqnarray}
\end{defi}

This operation has the following properties:
\begin{eqnarray}
&&U^{\dagger\dagger}=U,\qquad (UV)^\dagger=V^\dagger U^\dagger,\qquad (\lambda U+\mu V)^\dagger=\bar{\lambda} U^\dagger+\bar{\mu} V^\dagger\nonumber\\
&& U, V\in\C\otimes\cl_{p,q},\qquad \lambda, \mu\in\C.\nonumber
\end{eqnarray}

\begin{theorem} \cite{msh1}, \cite{msh3} The operation $U, V\in\C\otimes\cl_n \to(U,V):=\Tr(U^\dagger V)$ is a Hermitian (or Euclidian) scalar product on $\C\otimes\cl_{p,q}$ (or $\cl_{p,q}$ respectively).
\end{theorem}

\begin{proof}
The proof is similar to the proof of the previous theorem. Now we have
$(e_{i_1}\cdots e_{i_k}, e_{i_1}\cdots e_{i_k})=\Tr(e_{i_k}^{-1}\ldots e_{i_1}^{-1}e_{i_1}\cdots e_{i_k})=\Tr(e)=1$.
\end{proof}

Note that the Hermitian conjugation in the case of the real Clifford algebra $\cl_{p,q}$ is called transposition anti-involution. It is considered in \cite{Abl1}, \cite{Abl2}, \cite{Abl3} in more details.

We have the following relation between the Hermitian conjugation of Clifford algebra elements and the Hermitian conjugation of matrices.

\begin{theorem} \cite{msh1}, \cite{msh3}
We have \,$\gamma(U^\dagger)=(\gamma(U))^\dagger$,\, where
$$\gamma: \C\otimes\cl_{p,q}\to\left\{
                     \begin{array}{ll}
                       \Mat(2^{\frac{n}{2}},\C), & \hbox{if $n$ is even} \\
                       \Mat(2^{\frac{n-1}{2}},\C)\oplus\Mat(2^{\frac{n-1}{2}},\C), & \hbox{if $n$ is odd}
                     \end{array}
                   \right.$$
is faithful matrix representation of $\C\otimes\cl_{p,q}$ such that  $(\gamma(e_a))^{-1}=(\gamma(e_a))^\dagger.$
\end{theorem}

Let us consider the following Lie group in Clifford algebra
\begin{eqnarray}
\U\cl_{p,q}:=\{U\in \C\otimes\cl_{p,q} \semicolon U^\dagger U=e\}\cong \left\{
\begin{array}{ll} \U(2^{\frac{n}{2}}), & \hbox{if $n$ is even} \\
\U(2^{\frac{n-1}{2}})\oplus\U(2^{\frac{n-1}{2}}), & \hbox{if $n$ is odd.}
\end{array}\right.\nonumber
\end{eqnarray}
We call it \emph{unitary group in Clifford algebra} \cite{msh1}, \cite{msh3}. All basis elements of Clifford algebra lie in this group by definition $e_{a_1 \ldots a_k}\in\U\cl_{p,q}$. The corresponding Lie algebra is
$$
\u\cl_{p,q}:=\{U\in \C\otimes\cl_{p,q} \semicolon U^\dagger=-U\}.
$$

\begin{theorem}\label{thHerm} \cite{msh1}, \cite{msh3} We have the following formulas which can be considered as another (equivalent to Definition \ref{defHerm}) definition of Hermitian conjugation:
\begin{eqnarray}
\!\!\!U^\dagger=\left\{
            \begin{array}{ll}
             \!\! (e_{1\ldots p})^{-1} \overline{\widetilde{U}} e_{1\ldots p}, & \hbox{if $p$ is odd} \\
             \!\! (e_{1\ldots p})^{-1} \overline{\widehat{\widetilde{U}}} e_{1\ldots p}, & \hbox{if $p$ is even}
            \end{array}
          \right.
U^\dagger=\left\{
            \begin{array}{ll}
             \!\! (e_{p+1\ldots n})^{-1} \overline{\widetilde{U}} e_{p+1\ldots n}, & \hbox{if $q$ is even} \\
             \!\! (e_{p+1\ldots n})^{-1} \overline{\widehat{\widetilde{U}}}e_{p+1\ldots n}, & \hbox{if $q$ is odd.}
            \end{array}
          \right.\nonumber
\end{eqnarray}
\end{theorem}

As an example, we obtain well-known relations $\gamma_a^\dagger=\gamma_0 \gamma_a \gamma_0$ for the Dirac gamma-matrices.

\begin{proof}
Because of the linearity of the operation ${}^\dagger$ it is sufficient to prove the following formulas:
\begin{eqnarray}
{e_{i_1 \ldots i_k}}^\dagger=(-1)^{(p+1)k} e_{1\ldots p}^{-1} \widetilde{e_{i_1\ldots i_k}} e_{1\ldots p},\quad
{e_{i_1\ldots i_k}}^\dagger=(-1)^{qk} e_{p+1\ldots n}^{-1} \widetilde{e_{i_1\ldots i_k}} e_{p+1\ldots n}.\nonumber
\end{eqnarray}
Let $s$ be the number of common indices in $\{i_1, \ldots, i_k\}$ and $\{1, \ldots, p\}$.
Then
\begin{eqnarray}
&&(-1)^{(p+1)k} e_{1\ldots p}^{-1} \widetilde{e_{i_1\ldots i_k}} e_{1\ldots p}=(-1)^{(p+1)k} e_p \cdots e_1 e_{i_k} \cdots e_{i_1}e_1\cdots e_p\nonumber\\
&&=(-1)^{(p+1)k} (-1)^{kp-s} e_{i_k}\cdots e_{i_1}=(-1)^{k-s} e_{i_k}\cdots e_{i_1}=e_{i_1\dots i_k}^{-1}\nonumber\\
&&(-1)^{qk} e_{p+1\ldots n}^{-1} \widetilde{e_{i_1\ldots i_k}} e_{p+1\ldots n}=(-1)^{qk}(-1)^q e_n \cdots e_{p+1}e_{i_k}\cdots e_{i_1}e_{p+1}\cdots e_n\nonumber\\
&&=(-1)^{qk+q}(-1)^{qk-(k-s)}(-1)^q e_{i_k}\cdots e_{i_1}= (-1)^{k-s}e_{i_k}\cdots e_{i_1}=e_{i_1\ldots i_k}^{-1}.\nonumber
\end{eqnarray}
The theorem is proved. \end{proof}

\subsection{Primitive Idempotents and Minimal Left Ideals}\label{sec:3.5}

\begin{defi}
The element $t\in\C\otimes\cl_{p,q}$, $t^2=t$, $t^\dagger=t$ is called a Hermitian idempotent.
The subset $\I(t)=\{U\in\C\otimes\cl_{p,q} \semicolon U=Ut\}$ is called the left ideal generated by $t$.
\end{defi}

\begin{defi}
A left ideal that does not contain other left ideals except itself and the trivial ideal (generated by $t=0$) is called a minimal left ideal. The corresponding idempotent is called primitive.
\end{defi}

Note that if $V\in\I(t)$ and $U\in\C\otimes\cl_{p,q}$, then $UV\in \I(t)$.

The left ideal $\I(t)$ is a complex vector space with the orthonormal basis $\tau_1, \ldots, \tau_d$, $d:=\dim \I(t)$. We have the Hermitian scalar product $(U, V)=\Tr(U^\dagger V)$ on $\I(t)$, $\tau_k=\tau^k$, $(\tau_k, \tau^l)=\delta_k^l$, $k, l=1, \ldots, n$. We may define the linear map $\gamma: \C\otimes\cl_{p,q}\to \Mat(d,\C)$
\begin{eqnarray}
U\tau_k=\gamma(U)^l_k \tau_l,\qquad \gamma(U)=||\gamma(U)^l_k||\in\Mat(d,\C).\label{repres}
\end{eqnarray}
We have $\gamma(U)^k_l=(\tau^k, U\tau_l)$.

\begin{lemma}
The linear map $\gamma$ is a representation of Clifford algebra of the dimension $d$: $\gamma(UV)=\gamma(U)\gamma(V)$. \end{lemma}
\begin{proof}
$\gamma(UV)^m_k \tau_m=(UV)\tau_k=U(V\tau_k)=U \tau_l \gamma(V)^l_k=\gamma(U)^m_l \gamma(V)^l_k \tau_m.$
\end{proof}
\begin{lemma}
We have $\gamma(U^\dagger)=(\gamma(U))^\dagger$.
\end{lemma}
\begin{proof}
Using $(A, UB)=(AU^\dagger, B)$ and $(A,B)=\overline{(B,A)}$ for $(A, B)=\Tr(A^\dagger B)$, we obtain
$\gamma(U)^k_l=(U^\dagger \tau^k, \tau_l)$, $\overline{\gamma(U)}^k_l=(\tau_l, U^\dagger \tau^k)$. Transposing, we get $(\gamma(U)^k_l)^\dagger=(\tau^k, U^\dagger \tau_l)$, which coincides with $\gamma(U^\dagger)^k_l=(\tau^k, U^\dagger \tau_l)$.
\end{proof}

\begin{theorem} \cite{msh1}, \cite{msh3}
The following elements are primitive idempotents in $\C\otimes\cl_{p,q}$:
\begin{eqnarray}
&&t=\frac{1}{2}(e+i^a e_1) \prod_{k=1}^{[n/2]-1}\frac{1}{2}(e+i^{b_k} e_{2k}e_{2k+1})\in\C\otimes\cl_{p,q},\quad t^2=t^\dagger=t\nonumber\\
&&a=\left\{
    \begin{array}{ll}
      0, & \hbox{if $p\neq 0$} \\
      1, & \hbox{if $p=0$}
    \end{array}
  \right.\qquad
b_k=\left\{
      \begin{array}{ll}
        0, & \hbox{$2k=p$} \\
        1, & \hbox{$2k\neq p$.}
      \end{array}
    \right.\nonumber
\end{eqnarray}
\end{theorem}

Further, we choose a basis of the corresponding minimal left ideal $\I(t)$ (for mote details, see \cite{msh1}, \cite{msh3}) and obtain the representation of the complexified Clifford algebra (\ref{repres}).

For the real Clifford algebras $\cl_{p,q}$ the construction is similar, see \cite{ablam}. Using the idempotent and the basis of the left ideal, we can construct representations of the real Clifford algebra.

\section{Method of Averaging in Clifford Algebras}\label{sec:5}

\subsection{Averaging in Clifford Algebras}\label{sec:5.1}

The method of averaging in Clifford algebras is related to the method of averaging in the representation theory of finite groups \cite{Serre}, \cite{Dixon}, \cite{Babai}. We present a number of theorems which one can find in \cite{sh18}, \cite{sh14}, \cite{msh3}, \cite{msh1}, \cite{msh5}.

Let us consider the Reynolds operator \cite{Cox} of the Salingaros group \cite{Sal1}, \cite{Sal2}, \cite{Sal3} $\G_{p,q}:=\{\pm e_A\}$:
$$F(U)=\frac{1}{|\G_{p,q}|}\sum_{g\in\G_{p,q}}g^{-1}Ug=\frac{1}{2^n}\sum_A (e_A)^{-1} U e_A,\qquad  U\in\cl_{p,q}.$$

\begin{theorem}\cite{sh18} The operator $F(U)$ is the projection onto the center of Clifford algebra $\cl_{p,q}$:
$$\F(U)=\frac{1}{2^n}\sum_A e_A^{-1} U e_A=\left\lbrace
\begin{array}{ll}
<\!U\!>_0, & \mbox{if $n$ is even} \\
<\!U\!>_0+<\!U\!>_n, & \mbox{if $n$ is odd,}
\end{array}
\right.\qquad \F^2=\F.\nonumber
$$
\end{theorem}

Let us consider other operators acting on $U\in\cl_{p,q}$ that are also related to projection operators.

\begin{theorem}\cite{sh14} We have
\begin{eqnarray}
&&\F_{{\rm Even}}(U):=\!\!\!\!\!\sum_{A: |A|=0\!\!\!\!\mod 2}\!\!\!\!\! e_A^{-1} U e_A=<\!U\!>_0+<\!U\!>_n,\quad \F_{{\rm Even}}^2=\F_{{\rm Even}}\nonumber\\
&&\F_{{\rm Odd}}(U):=\!\!\!\!\!\sum_{A: |A|=1\!\!\!\!\mod 2}\!\!\!\!\! e_A^{-1} U e_A=<\!U\!>_0+(-1)^{n+1}<\!U\!>_n,\quad \F_{{\rm Odd}}^2=\F_{{\rm Odd}}.\nonumber
\end{eqnarray}
We have $\F=\frac{1}{2}(\F_{{\rm Even}}+\F_{{\rm Odd}})$ in the case of even $n$ and $\F=\F_{{\rm Even}}=\F_{{\rm Odd}}$ in the case of odd $n$.
\end{theorem}

\begin{theorem}\cite{sh14} For $m=0, 1, \ldots, n$ we have
$$F_m(U):=\sum_{A: |A|=m} e_A^{-1} U e_A=\sum_{k=0}^n (-1)^{km}(\sum_{i=0}^m (-1)^i C_k^i C_{n-k}^{m-i})<\!U\!>_k.$$
In particular case,
\begin{eqnarray}
F_1(U):=\sum e_a^{-1} U e_a = \sum_{k=0}^n (-1)^k (n-2k)<\!U\!>_k.\nonumber
\end{eqnarray}
\end{theorem}

\begin{theorem}\cite{msh5} Let us consider the operator $F_1$ from the previous theorem which acts several times: $F_1^l(U)=\underbrace{F_1(F_1(\cdots F_1(}_{l}U))\cdots)$, where $F_1^0(U)=U$. Then
\begin{eqnarray}
&&\mbox{If $n=p+q$ is even, then} \quad <\!U\!>_k=\sum_{l=0}^n b_{kl} F_1^l(U),\quad \mbox{where}\nonumber\\
&&B_{n+1}=||b_{kl}||=A_{n+1}^{-1},\, A_{n+1}=||a_{kl}||,\, a_{kl}=\lambda_{l-1}^{k-1},\, \lambda_{k}=(-1)^k (n-2k).\nonumber\\
&&\mbox{If $n=p+q$ is odd, then} \quad <\!U\!>_k+<\!U\!>_{n-k}=\sum_{l=0}^{\frac{n-1}{2}} g_{kl} F_1^l(U),\quad \mbox{where}\nonumber\\
&&G_{\frac{n+1}{2}}=||g_{kl}||=D_{\frac{n+1}{2}}^{-1},\, D_{\frac{n+1}{2}}=||d_{kl}||,\, d_{kl}=\lambda_{l-1}^{k-1}.\nonumber
\end{eqnarray}
\end{theorem}

\begin{theorem}\cite{sh14} For $k=1, \ldots, n-1$ we have:
\begin{eqnarray*}
\sum_{A: |A|=0\!\!\!\!\mod 4}\!\!\!\!\!\! e_A^{-1} <\!U\!>_k e_A&=&2^{\frac{n-2}{2}}\cos(\frac{\pi k}{2}-\frac{\pi n}{4})<\!U\!>_k\nonumber\\
\sum_{A: |A|=1\!\!\!\!\mod 4}\!\!\!\!\!\! e_A^{-1} <\!U\!>_k e_A&=&(-1)^{k+1}2^{\frac{n-2}{2}}\sin(\frac{\pi k}{2}-\frac{\pi n}{4})<\!U\!>_k\nonumber\\
\sum_{A: |A|=2\!\!\!\!\mod 4}\!\!\!\!\!\! e_A^{-1} <\!U\!>_k e_A&=&-2^{\frac{n-2}{2}}\cos(\frac{\pi k}{2}-\frac{\pi n}{4})<\!U\!>_k\nonumber\\
\sum_{A: |A|=3\!\!\!\!\mod 4}\!\!\!\!\!\! e_A^{-1} <\!U\!>_k e_A&=&(-1)^{k}2^{\frac{n-2}{2}}\sin(\frac{\pi k}{2}-\frac{\pi n}{4})<\!U\!>_k.\nonumber
\end{eqnarray*}
For $m=0, 1, 2, 3$ we have:
\begin{equation}
\sum_{A: |A|=m\!\!\!\!\mod 4}\!\!\!\!\!\!\!\!\!\! e_A^{-1}e_A=d_m(n)e,\, \sum_{A: |A|=m\!\!\!\!\mod 4}\!\!\!\!\!\!\!\!\!\! e_A^{-1} e_{1\ldots n} e_A=(-1)^{m(n+1)} d_m(n)e_{1\ldots n}\nonumber
\end{equation}
where $d_m(n)=\dim \overline{\textbf{m}}$ (see (\ref{dimquat})).
\end{theorem}

\begin{theorem}\cite{sh18} Let $\M_n$ be the matrix of the size $2^n$ with the elements $m_{AB}=e_A e_B e_A^{-1} e_B^{-1}|_{e\to 1}$ (it is the commutator of $e_A$ and $e_B$ in the Salingaros group). Then we have
$$\F_{e_A}(U):=e_A^{-1}Ue_A=\sum_B m_{AB} <\!U\!>_{e_B},\qquad U\in\cl_{p,q}$$
where $<\!U\!>_{e_B}$ is the projection of the element $U$ onto the subspace spanned over~$e_B$.
\end{theorem}

Using previous theorems, we can solve several classes of commutator equations (see \cite{sh18}, \cite{sh14})
$$e_A X+\epsilon X e_A=Q_A,\qquad \epsilon \in \R\setminus \{0\}, \qquad A\in \G$$
for some known elements $Q_A\in\cl_{p,q}$ and unknown element $X\in\cl_{p,q}$, where $\G$ is some subset of the set of all ordered multi-indices with a length between $0$ and~$n$.

One can find other properties of considered operators in \cite{sh18}, \cite{sh14}, \cite{msh3}.

\subsection{Pauli's Fundamental Theorem, Faithful and Irreducible Representations}\label{sec:5.2}

Let the set of Clifford algebra elements satisfies the conditions
\begin{eqnarray}
\{\beta_a \semicolon a=1, \ldots, n\}\in\cl_{p,q},\qquad \beta_a \beta_b + \beta_b \beta_a= 2 \eta_{ab} e.\label{yy1}
\end{eqnarray}
Then the set
\begin{eqnarray}
\gamma_a=T^{-1}\beta_a T\label{yy2}
\end{eqnarray}
for any invertible $T\in\cl_{p,q}$ satisfies the conditions
\begin{eqnarray}
\gamma_a \gamma_b + \gamma_b \gamma_a=2 \eta_{ab} e.\label{yy3}
\end{eqnarray}

Really,
$$\gamma_a\gamma_b+\gamma_b\gamma_a=T^{-1}\beta_a T T^{-1} \beta_b T + T^{-1} \beta_b T T^{-1} \beta_a T$$
$$=T^{-1} (\beta_a \beta_b + \beta_b \beta_a)T= T^{-1} 2\eta_{ab}e T=2\eta_{ab}e.$$

But we are interested in another question. Does the element $T$ (\ref{yy2}) exist for every two sets (\ref{yy3}) and (\ref{yy1})? W.~Pauli proved the following theorem in 1936.

\begin{theorem}[Pauli] \cite{Pauli}
Consider two sets of square complex matrices
\begin{eqnarray}\gamma_a,\qquad \beta_a,\qquad a=1, 2, 3, 4\nonumber\end{eqnarray}
of size $4$. Let these 2 sets satisfy the following conditions
\begin{eqnarray}
\gamma_a \gamma_b + \gamma_b \gamma_a&=& 2 \eta_{ab} {\bf1},\qquad \eta=\diag(1, -1, -1, -1)\nonumber \\
\beta_a \beta_b + \beta_b \beta_a&=& 2 \eta_{ab} {\bf1}.\nonumber
\end{eqnarray}
Then there exists a unique (up to multiplication by a complex constant) complex matrix $T$ such that
\begin{eqnarray}
\gamma_{a}=T^{-1}\beta_a T,\qquad a=1, 2, 3, 4.\nonumber
\end{eqnarray}
\end{theorem}

This theorem states that the complexified Clifford alegbra $\C\otimes\cl_{1,3}$ has unique (up to equivalence) faithful and irreducible representation of dimension $4$ .

Using the modern representation theory, we can obtain the following facts:
\begin{itemize}
  \item In the case of even $n=p+q$, $\C\otimes\cl_{p,q}$  has one faithful and irreducible representation of dimension $2^{\frac{n}{2}}$ ($\C\otimes\cl_{p,q}\cong\Mat(2^{\frac{n}{2}},\C)$, $n$ is even).
  \item In the case of odd $n=p+q$, $\C\otimes\cl_{p,q}$ has two irreducible representations of dimension $2^{\frac{n-1}{2}}$.
  \item In the case of odd $n=p+q$, $\C\otimes\cl_{p,q}$ has two faithful reducible representation of dimension $2^{\frac{n-1}{2}}+2^{\frac{n-1}{2}}=2^{\frac{n+1}{2}}$ ($\C\otimes\cl_{p,q}\cong\Mat(2^{\frac{n-1}{2}},\C)\oplus\Mat(2^{\frac{n-1}{2}},\C)$, $n$ is odd).
\end{itemize}

Similarly we can formulate statements for the real Clifford algebra $\cl_{p,q}$. The results depend on $n\mod 2$ and $p-q\mod 8$.

We also want to obtain an algorithm to compute the element $T$ that connects two sets of Clifford algebra elements. We can do this using the method of averaging in Clifford algebra and the operators $\sum_{A\in \G}\beta_A U \gamma_A^{-1}$, where $\G$ is some subset of the set of all ordered multi-indices with a length between $0$ and $n$. One can find different properties of these operators in \cite{sh17}.

We have the following theorems.

\begin{theorem}[The case of even $n$]\label{thPauli1}\cite{sh3} Consider the real $\cl_{p,q}$ (or the complexified $\C\otimes\cl_{p,q}$) Clifford algebra with even $n=p+q$. Let two sets of Clifford algebra elements
$\gamma_a,\, \beta_a,\, a=1, 2, \ldots, n$ satisfy conditions
$$\gamma_a \gamma_b + \gamma_b \gamma_a= 2 \eta_{ab} e,\quad \beta_a \beta_b + \beta_b \beta_a= 2 \eta_{ab} e.$$

Then both sets generate bases of Clifford algebra and there exists an unique (up to multiplication by a real (respectively complex) number) Clifford algebra element $T$ such that
$$
\gamma_{a}=T^{-1}\beta_a T,\qquad a=1, \ldots, n.
$$

Additionally, we can obtain this element $T$ in the following way
$$T=H(F):=\frac{1}{2^n}\sum_{A} \beta_A F (\gamma_A)^{-1}$$
where $F$ is an element of a set
\begin{eqnarray}
&&1)\, \{ \gamma_A \semicolon |A|=0\!\!\!\mod 2\} \quad\mbox{if}\, \beta_{1\ldots n}\neq-\gamma_{1\ldots n}\nonumber\\
&&2)\, \{ \gamma_A \semicolon |A|=1\!\!\!\mod 2\} \quad\mbox{if}\, \beta_{1\ldots n}\neq\gamma_{1\ldots n}
\nonumber
\end{eqnarray}
such that $H(F)\neq 0$.
\end{theorem}

Let us consider the case of odd $n$. We start with two examples.

\begin{ex} Let us consider the Clifford algebra $\cl_{2,1}\simeq \Mat(2,\R)\oplus\Mat(2,\R)$ with the generators $e_1, e_2, e_3$. We can take
$$\gamma_1=e_1,\qquad \gamma_2=e_2,\qquad \gamma_3=e_1 e_2.$$
Then  $\gamma_a \gamma_b + \gamma_b \gamma_a= 2 \eta_{ab} {\bf1}$. The elements $\gamma_1, \gamma_2, \gamma_3$ generate not $\cl_{2,1}$, but $\cl_{2,0}\simeq\Mat(2,\R)$.
\end{ex}

\begin{ex} Let us consider the Clifford algebra $\cl_{3,0}\simeq\Mat(2,\C)$ with the generators $e_1, e_2, e_3$. We can take
$$\beta_1=\sigma_1=\left( \begin{array}{ll}
 0 & 1 \\
 1 & 0 \end{array}\right),\quad \beta_2=\sigma_2=\left( \begin{array}{ll}
 0 & -i \\
 i & 0 \end{array}\right),\quad \beta_3=\sigma_3=\left( \begin{array}{ll}
 1 & 0 \\
 0 & -1 \end{array}\right)$$
$$\gamma_a=-\sigma_a,\qquad a=1, 2, 3.$$
Then $\gamma_{123}=-\beta_{123}$. Suppose that we have $T\in\GL(2, \C)$ such that $\gamma_a=T^{-1}\beta_a T$. Then
$$\gamma_{123}=T^{-1}\beta_1 T T^{-1}\beta_2 T T^{-1} \beta_3 T=T^{-1} \beta_1 \beta_2 \beta_3 T =\beta_{123}$$
and we obtain a contradiction (we use that $\beta_{123}=\sigma_{123}=i\left( \begin{array}{ll}
 1 & 0 \\
 0 & 1 \end{array}\right)=i {\bf 1}$).

But we have $T={\bf 1}$ such that $\gamma_a=-T^{-1}\beta_a T$.
\end{ex}

\begin{theorem}[The case of odd $n$]\label{thPauli2}\cite{sh3} Consider the real $\cl_{p,q}$ (or the complexified $\C\otimes\cl_{p,q}$) Clifford algebra with odd $n=p+q$. Suppose that two sets of Clifford algebra elements
$\gamma_a,\, \beta_a,\, a=1, 2, \ldots, n$
satisfy conditions
$$\gamma_a \gamma_b + \gamma_b \gamma_a= 2 \eta_{ab} e,\quad
\beta_a \beta_b + \beta_b \beta_a= 2 \eta_{ab} e.$$

Then, in the case of the Clifford algebra of signature $p-q\equiv1\!\!\mod4$, elements $\gamma_{1\ldots n}$ and $\beta_{1\ldots n}$ either take the values $\pm e_{1\ldots n}$ and the corresponding sets generate bases of Clifford algebra (and we have cases 1-2 below) or take the values $\pm e$ and then the sets do not generate bases (and we have cases 3-4 below).

In the case of the Clifford algebra of signature $p-q\equiv3\!\!\mod4$, elements $\gamma_{1\ldots n}$ and $\beta_{1\ldots n}$ either take the values $\pm e_{1\ldots n}$ and the corresponding sets generate bases of Clifford algebra (and we have cases 1-2 below) or take the values $\pm ie$ (this is possible only in the case of the complexified Clifford algebra) and then the sets do not generate bases (and we have cases 5-6 below).

There exists an unique (up to multiplication by an invertible element of the center of the Clifford algebra) element $T$ such that
\begin{eqnarray}
1)&& \gamma_{a}=T^{-1}\beta_a T,\qquad  a=1, \ldots, n\quad\Leftrightarrow\quad \beta_{1\ldots n}=\gamma_{1\ldots n}\nonumber\\
2)&& \gamma_{a}=-T^{-1}\beta_a T,\qquad  a=1, \ldots, n\quad\Leftrightarrow\quad \beta_{1\ldots n}=-\gamma_{1\ldots n}\nonumber\\
3)&& \gamma_{a}=e_{1\ldots n}T^{-1}\beta_a T,\qquad  a=1, \ldots, n\quad\Leftrightarrow\quad \beta_{1\ldots n}=e_{1\ldots n}\gamma_{1\ldots n}\nonumber\\
4)&&\gamma_{a}=-e_{1\ldots n}T^{-1}\beta_a T,\qquad  a=1, \ldots, n \quad\Leftrightarrow\quad \beta_{1\ldots n}=-e_{1\ldots n}\gamma_{1\ldots n}\nonumber\\
5)&& \gamma_{a}=i e_{1\ldots n}T^{-1}\beta_a T,\qquad  a=1, \ldots, n\quad\Leftrightarrow\quad \beta_{1\ldots n}=i e_{1\ldots n}\gamma_{1\ldots n}\nonumber\\
6)&&\gamma_{a}=-i e_{1\ldots n}T^{-1}\beta_a T,\qquad   a=1, \ldots, n \quad\Leftrightarrow\quad \beta_{1\ldots n}=-i e_{1\ldots n}\gamma_{1\ldots n}.\nonumber
\end{eqnarray}

Note that all six cases have the unified notation $\gamma_a=\beta_{1\ldots n}(\gamma_{1\ldots n})^{-1} T^{-1}\beta_a T.$

Additionally, in the case of the real Clifford algebra $\cl_{p,q}$ of signature $p-q\equiv 1 \mod 4$ and the complexified Clifoord algebra $\C\otimes\cl_{p,q}$ of arbitrary signature, the element $T$, whose existence is stated in cases 1-6 of the theorem, equals
\begin{eqnarray}T=H_{Even}(F):=\frac{1}{2^{n-1}}\sum_{A: |A|=0\!\!\!\mod 2}\beta_A F \gamma_A^{-1}\nonumber
\end{eqnarray}
where $F$ is an element of the set $\{\gamma_A+\gamma_B \semicolon |A|=0\!\!\!\mod 2, |B|=0\!\!\!\mod 2\}.$

In the case of the real Clifford algebra $\cl_{p,q}$ of signature $p-q\equiv 3 \mod4$, the element $T$, whose existence is stated in cases 1 and 2 of the theorem, equals $T=H_{Even}(F)$, where $F$ is an element of the set $\{\gamma_A \semicolon |A|=0\!\!\!\mod 2\}$ such that $H_{Even}(F)\neq 0$.
\end{theorem}

Using the algorithm to compute the element $T$ in Theorems \ref{thPauli1} and \ref{thPauli2}, we present an algorithm to compute elements of spin groups in \cite{sh15}.

In \cite{msh2}, we present a local variant of Pauli theorem, when two sets of Clifford algebra elements smoothly depend on the point of pseudo-Euclidian space.

\section{Lie Groups and Lie Algebras in Clifford Algebras}\label{sec:6}

\subsection{Orthogonal Groups}\label{sec:6.1}

Let us consider pseudo-orthogonal group $\O(p,q)$, $p+q=n$:
$$
\O(p,q):=\{A\in\Mat(n,\R) \semicolon A^\T \eta A=\eta\},\, \eta=\diag(\underbrace{1,\ldots , 1}_p,\underbrace{-1, \ldots, -1}_{q}).
$$
It can be proved that (for more details, see \cite{Lester} and \cite{msh3})
$$A\in\O(p,q)\Rightarrow \det A=\pm 1,\, |A^{1\ldots p}_{1\ldots p}|\geq 1,\, |A^{p+1 \ldots ,n}_{p+1 \ldots n}|\geq 1,\, A^{1\ldots p}_{1\ldots p}=\frac{A^{p+1 \ldots n}_{p+1\ldots n}}{\det A}$$
where $A^{1\ldots p}_{1\ldots p}$ and $A^{p+1 \ldots ,n}_{p+1 \ldots n}$ are the minors of the matrix $A$.
The group $\O(p,q)$ has the following subgroups:
\begin{eqnarray}
&&\SO(p,q):=\{A\in\O(p,q)\semicolon \det A=1\}\nonumber\\
&&\SO_+(p,q):=\{A\in\SO(p,q) \semicolon A^{1\ldots p}_{1\ldots p}\geq 1\}=\{A\in\SO(p,q) \semicolon A^{p+1\ldots n}_{p+1\ldots n}\geq 1\}\nonumber\\
&&= \{A\in\O(p,q) \semicolon A^{1\ldots p}_{1\ldots p}\geq 1, A^{p+1\ldots n}_{p+1\ldots n}\geq 1\}\nonumber\\
&&\O_+(p,q):=\{A\in\O(p,q) \semicolon A^{1\ldots p}_{1\ldots p}\geq 1\}\nonumber\\
&&\O_-(p,q):=\{A\in\O(p,q) \semicolon A^{p+1\ldots n}_{p+1\ldots n}\geq 1\}.\nonumber
\end{eqnarray}
The group $\O(p,q)$ has four components in the case $p\neq 0$, $\neq 0$:
$$\O(p,q)=\SO_+(p,q)\sqcup\O_+(p,q)'\sqcup\O_-(p,q)'\sqcup\SO(p,q)'$$ $$\O_+(p,q)=\SO_+(p,q)\sqcup\O_+(p,q)',\quad \O_-(p,q)=\SO_+(p,q)\sqcup\O_-(p,q)'$$
$$\SO(p,q)=\SO_+(p,q)\sqcup\SO(p,q)'.$$

\begin{ex} In the cases $p=0$ or $q=0$, we have orthogonal groups $\O(n):=\O(n,0)\cong\O(0,n)$, special orthogonal groups $\SO(n):=\SO(n,0)\cong\SO(0,n)$. The group $\O(n)$ has two connected components: $\O(n)=\SO(n)\sqcup\SO(n)'$.
\end{ex}
\begin{ex} In the case $(p,q)=(1,3)$, we have Lorentz group $\O(1,3)$, special (or proper) Lorentz group $\SO(1,3)$, orthochronous Lorentz group $\O_+(1,3)$, orthochorous (or parity preserving) Lorentz group $\O_-(1,3)$, proper orthochronous Lorentz group $\SO_+$(1,3).
\end{ex}

\begin{defi}
A subgroup $H\subset G$ of a group $G$ is called a normal subgroup ($H\triangleleft G$) if $g H g^{-1}\subseteq H$ for all $g\in G$.
\end{defi}
\begin{defi} A quotient group (or factor group) $\frac{G}{H}:=\{ gH \semicolon g\in G\}$ is the set of all left cosets ($\equiv$ right cosets, because $H$ is normal).
\end{defi}

All considered subgroups are normal (for example, $\SO_+(p,q)\triangleleft \O(p,q)$) and
\begin{eqnarray}
&&\frac{\O(p,q)}{\SO_+(p,q)}=\Z_2\times \Z_2,\qquad \frac{\O(n)}{\SO(n)}=\Z_2,\qquad \frac{\O(p,q)}{\SO(p,q)}=\frac{\O(p,q)}{\O_-(p,q)}\nonumber\\
&&=\frac{\O(p,q)}{\O_+(p,q)}=
\frac{\SO(p,q)}{\SO_+(p,q)}=\frac{\O_-(p,q)}{\SO_+(p,q)}= \frac{\O_+(p,q)}{\SO_+(p,q)}=\Z_2.\label{quot}
\end{eqnarray}

\begin{ex}
The group $\O(1,1)$ has four connected components $\O_+'(1,1)$, $\O_-'(1,1)$, $\SO'(1,1)$, $\SO_+(1,1)$ of matrices of the following type respectively (note that $\cosh^2\psi=1+\sinh^2\psi$ and $\cosh\psi\geq1$):
\begin{eqnarray}
&&\left(\begin{array}{ll}
 \!\!\cosh\psi & \!\!\sinh\psi \\
 \!\!-\sinh\psi & \!\!-\cosh\psi \end{array}\right),
\left(\begin{array}{ll}
 \!\!-\cosh\psi & \!\!-\sinh\psi \\
 \!\!\sinh\psi & \!\!\cosh\psi \end{array}\right)\nonumber\\
&&\left(\begin{array}{ll}
 \!\!-\cosh\psi & \!\!-\sinh\psi \\
 \!\!-\sinh\psi & \!\!-\cosh\psi \end{array}\right),
\left(\begin{array}{ll}
 \!\!\cosh\psi & \!\!\sinh\psi \\
 \!\!\sinh\psi & \!\!\cosh\psi \end{array}\right),\quad \psi\in\R.\nonumber
\end{eqnarray}
\end{ex}

\subsection{Lipschitz and Clifford Groups}\label{sec:6.2}

Let us consider \emph{the group of all invertible elements}
$$\cl^\times_{p,q}:=\{U\in\cl_{p,q} \semicolon V\in\cl_{p,q} \, \mbox{exists}: UV=VU=e \}$$
of dimension $\dim\cl^\times_{p,q}=2^n$. The corresponding Lie algebra is $\cl_{p,q}$ with the Lie bracket $[U,V]=UV-VU$.

Let us consider the \emph{adjoint representation}
$$\Ad: \cl^\times_{p,q}\to \Aut\cl_{p,q},\, T \to \Ad_T,\, \Ad_T U=TUT^{-1},\, U\in\cl_{p,q}.$$
The kernel of $\Ad$ is (see Theorem \ref{thcenter})
$$\ker(\Ad)=\{T\in\cl^{\times}_{p,q} \semicolon \Ad_T(U)=U \quad \mbox{for all}\,\, U\in\cl_{p,q}\}$$
$$=\left\{                                                                                                    \begin{array}{ll}                                                                                             \cl^{0\times}_{p,q}, & \hbox{if $n$ is even} \\                                                                                                     (\cl^0_{p,q}\oplus\cl^n_{p,q})^\times, & \hbox{if $n$ is odd.}                                                                                                    \end{array}                                                                                                  \right.
$$

Let us consider the \emph{twisted adjoint representation}
$$\widetilde{\Ad}: \cl^\times_{p,q}\to \End\cl_{p,q},\, T \to \widetilde{\Ad_T},\, \widetilde{\Ad_T} U=\widehat{T}UT^{-1},\, U\in\cl_{p,q}.$$
The kernel of $\widetilde{\Ad}$ is
$$\ker(\widetilde{\Ad})=\{T\in\cl^{\times}_{p,q} \semicolon \widetilde{\Ad}_T(U)=U \quad \mbox{for all}\,\, U\in\cl_{p,q}\}=\cl^{0 \times}_{p,q}.$$

In the Clifford algebra $\cl_{p,q}$, we have a vector subspace $V=\cl^1_{p,q}$ with a quadratic form $Q(x)$ or a symmetric bilinear form $g(x,x)$:
$$g(x,y)=\frac{1}{2}(Q(x+y)-Q(x)-Q(y))=\frac{1}{2}(xy+yx)|_{e\to 1},\, x, y\in\cl^1_{p,q}.$$

\begin{lemma}
$\widetilde{\Ad}: \cl^{1\times}_{p,q} \to \O(p,q)$ on $V$.
\end{lemma}
\begin{proof}
For $v\in\cl^{1\times}_{p,q}$ and $x\in\cl^1_{p,q}$ we have
$$Q(\widetilde{\Ad}_v x)=(\hat{v}xv^{-1})^2=\hat{v} x v^{-1} \hat{v} x v^{-1}=x^2=Q(x)$$
because $x^2\in\cl^0_{p,q}$.
\end{proof}

$\widetilde{\Ad}_v$ acts on $V$ as a reflection along $v$ (in the hyperplane orthogonal to $v$):
$$\widetilde{\Ad}_v x=\hat{v} x v^{-1}= x-(xv+vx)v^{-1}=x-2 \frac{g(x,v)}{g(v,v)}v,\quad v\in \cl^{1\times}_{p,q},\quad x\in\cl^1_{p,q}.$$

\begin{theorem}[Cartan-Diedonn\'{e}]\label{thCD} Every orthogonal transformation on a nongenerate space $(V, g)$ is a product of reflections (the number $\leq \dim V$) in hyperplanes.
\end{theorem}

Let us consider the group  $\Gamma^2_{p,q}:=\{v_1 v_2 \cdots v_k \semicolon v_1, \ldots, v_k\in\cl^{1\times}_{p,q}\}.$

\begin{lemma}\label{lemmagamma2}
$\widetilde{\Ad}(\Gamma^2_{p,q})=\O(p,q)$ (surjectivity).
\end{lemma}
\begin{proof} If $f\in\O(p,q)$, then
\begin{eqnarray}
f(x)&=&\widetilde{\Ad}_{v_1}\circ\cdots\circ\widetilde{\Ad}_{v_k}(x)=\hat{v_1}\cdots\hat{v_k}x v_k^{-1}\cdots v_1^{-1}\nonumber\\
&=&\widehat{v_1 \cdots v_k}x (v_1 \cdots v_k)^{-1}=\widetilde{\Ad}_{v_1\cdots v_k}(x)\nonumber
\end{eqnarray}
 for $v_1, \ldots, v_k\in V^\times$ and $x\in V$.
\end{proof}

Let us consider the group $\Gamma^1_{p,q}:=\{T\in\cl^\times_{p,q} \semicolon \hat{T} x T^{-1}\in\cl^1_{p,q}\, \mbox{for all}\, x\in\cl^1_{p,q}\}$ and the
\emph{norm mapping (norm function)} $N: \cl_{p,q} \to \cl_{p,q}$, $N(U):=\widehat{\widetilde{U}} U$.

\begin{lemma}\label{lemmanorm}
$N: \Gamma^1_{p,q} \to \cl^{0 \times}_{p,q}$ $\cong\R^\times$.
\end{lemma}
\begin{proof} If $T\in\Gamma^1_{p,q}$ and $x\in\cl^1_{p,q}$, then
$$\widehat{T} x T^{-1}=\widetilde{\widehat{T} x T^{-1}}=\widetilde{T^{-1}} x \widetilde{\widehat{T}}=(\widetilde{T})^{-1}x \widetilde{\widehat{T}}.$$
Since $\widehat{\widehat{\widetilde{T}} T} x=x \widehat{\widetilde{T}} T$, it follows that $\widehat{\widetilde{T}} T\in\ker\widetilde{\Ad}=\cl^{0\times}_{p,q}$.
\end{proof}

\begin{lemma}
$N: \Gamma^1_{p,q}\to \R^\times$ is a group homomorphism:
$$N(UV)=N(U)N(V),\qquad N(U^{-1})=(N(U))^{-1},\qquad U, V\in\Gamma^1_{p,q}.$$
\end{lemma}

\begin{proof} We have
$$N(UV)=\widehat{\widetilde{UV}}UV=\widehat{\widetilde{V}}\widehat{\widetilde{U}}UV= \widehat{\widetilde{V}} N(U) V=N(U) N(V)$$ and
$$e=N(e)=N(U U^{-1})=N(U) N(U^{-1}).$$
\end{proof}

\begin{lemma}\label{lemmagamma1}
$\widetilde{\Ad}: \Gamma^1_{p,q} \to \O(p,q)$.
\end{lemma}

\begin{proof} We have
$$N(\widehat{T})=\widehat{\widetilde{\widehat{T}}} \widehat{T}=\widehat{\widehat{\widetilde{T}} T} =\widehat{N(T)}=N(T)$$
 and
\begin{eqnarray}
N(\widetilde{\Ad}_T(x))&=& N(\widehat{T}xT^{-1})=N(\widehat{T})N(x)N(T^{-1})\nonumber\\
&=& N(T)N(x)(N(T))^{-1}=N(x).\nonumber
 \end{eqnarray}
 Since $N(x)=\widehat{\widetilde{x}}x=-x^2=-Q(x)$, it follows that $Q(\widetilde{\Ad}_T(x))=Q(x)$.
\end{proof}

\begin{lemma}\label{lemmagamma12}
$\Gamma^1_{p,q}=\Gamma^2_{p,q}$.
\end{lemma}

\begin{proof}
We know that $\Gamma^2_{p,q}\subseteq \Gamma^1_{p,q}$. Let us prove that $\Gamma^1_{p,q}\subseteq \Gamma^2_{p,q}$. If $T\in\Gamma^1_{p,q}$, then $\widetilde{\Ad}_T\in\O(p,q)$ by Lemma \ref{lemmagamma1}. Using Lemma \ref{lemmagamma2}, we conclude that $S\in\Gamma^2_{p,q}$ exists: $\widetilde{\Ad}_S=\widetilde{\Ad}_T$. We obtain $\widetilde{\Ad}_{T S^{-1}}=\id$ and $TS^{-1}=\lambda e$, $\lambda\in\R$. Finally, $T=\lambda S\in\Gamma^2_{p,q}$.
\end{proof}

\begin{defi} The following group is called Lipschitz group
\begin{eqnarray}
\Gamma^\pm_{p,q}:=\Gamma^1_{p,q}=\Gamma^2_{p,q}&=& \{T\in\cl^{(0)\times}_{p,q}\cup\cl^{(1)\times}_{p,q} \semicolon T x T^{-1}\in \cl^1_{p,q} \, \mbox{for all}\, x\in\cl^1_{p,q}\}\nonumber\\
&=&\{v_1 v_2 \cdots v_k \semicolon v_1, \ldots, v_k\in\cl^{1\times}_{p,q}\}.\nonumber
\end{eqnarray}
\end{defi}
\begin{defi} The following group is called Clifford group
$$\Gamma_{p,q}:=\{T\in\cl^\times_{p,q} \semicolon  T x T^{-1}\in \cl^1_{p,q}\, \mbox{for all} \, x\in\cl^1_{p,q}\}\supseteq \Gamma^\pm_{p,q}.$$
\end{defi}

So, we have $\widetilde{\Ad}(\Gamma^\pm_{p,q})=\O(p,q)$, i.e.
\begin{eqnarray}
\mbox{for any}\, P=||p^a_b||\in\O(p,q) \,\, \mbox{there exists}\,\, T\in\Gamma^\pm_{p,q}: \widehat{T} e_a T^{-1}=p_a^b e_b.\label{eqr1}
\end{eqnarray}

Let us consider the following subgroup of Clifford group
$$\Gamma^+_{p,q}:=\{T\in\cl^{(0)\times}_{p,q} \semicolon T x T^{-1}\in \cl^1_{p,q}\, \mbox{for all} \, x\in\cl^1_{p,q}\}\subset \Gamma^\pm_{p,q}.$$

We have
$\widetilde{\Ad}(\Gamma^+_{p,q})=\Ad(\Gamma^+_{p,q})=\SO(p,q)$, i.e.
\begin{eqnarray}
\mbox{for all} \, P=||p^a_b||\in\SO(p,q) \, \, \mbox{there exists}\,\, T\in\Gamma^+_{p,q}: \widehat{T} e_a T^{-1}=Te_a T^{-1}=p_a^b e_b.\label{eqr2}
\end{eqnarray}

We can prove (\ref{eqr1}) and (\ref{eqr2}) without the Cartan-Diedonn\'{e} theorem (see Theorem \ref{thCD} and Lemmas \ref{lemmagamma2} - \ref{lemmagamma12}) but with the use of the Pauli theorem (see Theorems \ref{thPauli1} and \ref{thPauli2}). One can find this approach in \cite{sh11} and \cite{msh3}.

\subsection{Spin Groups}\label{sec:6.3}

Let us define spin groups as normalized Lipschitz subgroups.
\begin{defi} The following groups are called spin groups:
\begin{eqnarray}
&&\Pin(p,q):=\{ T\in\Gamma^\pm_{p,q} \semicolon \widetilde{T} T=\pm e\}=\{ T\in\Gamma^\pm_{p,q} \semicolon \widehat{\widetilde{T}} T=\pm e\}\nonumber\\
&&\Pin_+(p,q):=\{T\in\Gamma^\pm_{p,q} \semicolon \widehat{\widetilde{T}} T=+e\}\nonumber\\
&&\Pin_-(p,q):=\{T\in\Gamma^\pm_{p,q} \semicolon \widetilde{T} T=+e\}\label{spingr}\\
&&\Spin(p,q):=\{T\in\Gamma^+_{p,q} \semicolon \widetilde{T} T= \pm e\}= \{T\in\Gamma^+_{p,q} \semicolon \widehat{\widetilde{T}} T= \pm e\}\nonumber\\
&&\Spin_+(p,q):=\{T\in\Gamma^+_{p,q} \semicolon \widetilde{T} T=+e\}=\{T\in\Gamma^+_{p,q} \semicolon \widehat{\widetilde{T}} T=+e\}.\nonumber
\end{eqnarray}
\end{defi}
In the case $p\neq 0$ and $q\neq 0$, we have
$$\Pin(p,q)=\Spin_+(p,q)\sqcup\Pin_+(p,q)'\sqcup\Pin_-(p,q)'\sqcup \Spin(p,q)'$$
$$\Pin_+(p,q)=\Spin_+(p,q)\sqcup\Pin_+(p,q)',\, \Pin_-(p,q)=\Spin_+(p,q)\sqcup\Pin_-(p,q)'$$ $$\Spin(p,q)=\Spin_+(p,q)\sqcup\Spin(p,q)'.$$

In Euclidian cases, we have two groups:
$$\Pin(n):=\Pin(n,0)=\Pin_-(0,n),\, \Spin(n,0)=\Pin_+(n,0)=\Spin_+(n,0)$$
$$\Pin(0,n):=\Pin(0,n)=\Pin_+(0,n),\, \Spin(0,n)=\Pin_-(0,n)=\Spin_+(0,n).$$

All considered subgroups are normal (for example, $\Spin_+(p,q)\triangleleft\Spin(p,q)$).

All quotient groups are the same as for the group $\O(p,q)$ and its subgroups respectively (see (\ref{quot})).

\begin{theorem}\label{thDC} The following homomorphisms are surjective with the kernel $\{\pm 1\}$:\\
\begin{eqnarray}
&&\widetilde{\Ad}: \Pin(p,q) \to \O(p,q)\nonumber\\
&&\widetilde{\Ad}: \Spin(p,q) \to \SO(p,q)\nonumber\\
&&\widetilde{\Ad}: \Spin_+(p,q) \to \SO_+(p,q)\nonumber\\
&&\widetilde{\Ad}: \Pin_+(p,q) \to \O_+(p,q)\nonumber\\
&&\widetilde{\Ad}: \Pin_-(p,q) \to \O_-(p,q).\nonumber
\end{eqnarray}
\end{theorem}
It means that
\begin{eqnarray}
\mbox{for all}\, P=||p^a_b||\in\O(p,q) \,\, \mbox{there exists}\,\, \pm T\in\Pin(p,q): \widehat{T} e_a T^{-1}=p_a^b e_b\label{PT}
\end{eqnarray}
and for the other groups similarly.
\begin{proof} Statement for the group $\Pin(p,q)$ follows from the statements of the previous section (see
Lemmas \ref{lemmagamma2} and \ref{lemmanorm}). For the other groups statement follows from the theorem on the norm of elements of spin groups which we give below (see \cite{sh4} and \cite{sh6}).
\end{proof}

\begin{theorem}\cite{sh4}, \cite{sh6} The square of the norm of the element $T\in\Pin(p,q)$ in (\ref{PT}) equals
$$||T||^2=\Tr(T^\dagger T)=
\left\lbrace
\begin{array}{ll}
P^{1\ldots p}_{1\ldots p}=P^{p+1 \ldots n}_{p+1 \ldots n}, & T\in\Spin_+(p,q)\\
P^{1\ldots p}_{1\ldots p}=-P^{p+1 \ldots n}_{p+1 \ldots n}, & T\in\Pin_+(p,q)'\\
-P^{1\ldots p}_{1\ldots p}=P^{p+1 \ldots n}_{p+1 \ldots n}, & T\in\Pin_-(p,q)'\\
-P^{1\ldots p}_{1\ldots p}=-P^{p+1 \ldots n}_{p+1 \ldots n}, & T\in\Spin(p,q)'
\end{array}
\right.\nonumber
$$
where $P^{1\ldots p}_{1\ldots p}$ and $P^{p+1 \ldots n}_{p+1 \ldots n}$ are the minors of the matrix $P\in\O(p,q)$ that corresponds to the element $T$ by (\ref{PT}).
\end{theorem}

\begin{theorem} We have the isomorphism $\Spin(p,q)\cong\Spin(q,p)$.
\end{theorem}
\begin{proof} This follows from the isomorphism $\cl^{(0)}_{p,q}\cong \cl^{(0)}_{q,p}$ (see Theorem \ref{thEven}).
\end{proof}

\begin{ex} We have
$\Spin(1,0)=\Spin(0,1)=\{\pm e\}=\Z_2$.
\end{ex}

\begin{ex} Note that $\Pin(p,q)\ncong \Pin(q,p)$ in general case. For example, $\Pin(1,0)=\{\pm e, \pm e_1\}\cong \Z_2\times \Z_2$ and $\Pin(0,1)\cong\Z_4$.
\end{ex}

\begin{theorem}
The condition $T x T^{-1}\in\cl^1_{p,q}$ for all $x\in\cl^1_{p,q}$ holds automatically in the cases $n\leq 5$ for all spin groups (\ref{spingr}), i.e.
$$\Pin(p,q)=\{T\in\cl^{(0)}_{p,q}\cup\cl^{(1)}_{p,q}  \semicolon \widetilde{T} T=\pm e\},\quad n=p+q\leq 5.$$
\end{theorem}

\begin{proof} If $T\in\cl^{(0)}_{p,q}\cup\cl^{(1)}_{p,q}$, then $TxT^{-1}\in\cl^1_{p,q}\oplus\cl^3_{p,q}\oplus\cl^5_{p,q}$. Using $\widetilde{T}T=\pm e$, we get $\widetilde{TxT^{-1}}=\widetilde{\pm T x \widetilde{T}}=\pm T x \widetilde{T}$ and $TxT^{-1}\in\cl^1_{p,q}\oplus\cl^5_{p,q}$. The statement is proved for $n\leq 4$.

Suppose that $n=5$ and $TxT^{-1}=v+\lambda e_{1\ldots 5}$, $v\in\cl^1_{p,q}$, $\lambda\in\R^\times$. Then $$\lambda=(TxT^{-1}e_{1\ldots 5}^{-1}-ve_{1\ldots 5}^{-1})|_{e\to 1}=\Tr(TxT^{-1}e_{1\ldots 5}^{-1})=\Tr(xe_{1\ldots 5}^{-1})=0$$
and we obtain a contradiction.
\end{proof}

\begin{ex} If the case $n=6$ the previous theorem is not valid. The element $T=\frac{1}{\sqrt{2}}(e_{12}+e_{3456})\in\cl^{(0)}_{6,0}$ satisfies $\widetilde{T}T=e$, but $Te_1 T^{-1}=-e_{23456}\notin\cl^1_{6,0}$.
\end{ex}

\begin{theorem} $\Spin_+(p,q)$ is isomorphic to the following groups in Table \ref{tab:5} in the cases $n=p+q\leq 6$.

\begin{table}[ht]
\centering
\begin{tabular}{|p{0.045\linewidth}|p{0.08\linewidth}|p{0.11\linewidth}| p{0.11\linewidth}|p{0.1\linewidth}|p{0.1\linewidth}|p{0.1\linewidth}|p{0.08\linewidth}|}\hline
$p \backslash q$ & 0 & 1 & 2 & 3 & 4 & 5 & 6   \\ \hline \hline
0 & $ \O(1)$& $\O(1)$ & $\U(1)$ & $\SU(2)$ & $^2\SU(2)$ & $\Sp(2)$ & $\SU(4)$ \\ \hline
1 & $\O(1)$& $\GL(1,\R)$& $\SU(1,1)$ & $\Sp(1,\C)$ & $\Sp(1,1)$ & $\SL(2, \H)$ &  \\ \hline
2 & $\U(1)$ & $\SU(1,1)$ & $^2\SU(1,1)$ & $\Sp(2,\R)$ & $\SU(2,2)$ &  &   \\ \hline
3 & $\SU(2)$ & $\Sp(1,\C)$ & $\Sp(2,\R)$ & $\SL(4,\R)$ &  &  &   \\ \hline
4 & $^2\SU(2)$& $\Sp(1,1)$ & $\SU(2,2)$ &  &  &  &   \\ \hline
5 & $\Sp(2)$ & $\SL(2, \H)$ &  &  &  &  &   \\ \hline
6 & $\SU(4)$ &  &  &  &  &  &   \\ \hline
\end{tabular}
\medskip
\caption{Isomorhisms between $\Spin_+(p,q)$ and matrix Lie groups}\label{tab:5}
\end{table}
\end{theorem}
Note that
\begin{eqnarray}
&&\U(1)\simeq\SO(2),\qquad \SU(2)\simeq \Sp(1),\qquad \SL(2,\C)\simeq \Sp(1,\C)\nonumber\\
&&\SU(1,1)\simeq \SL(2,\R)\simeq \Sp(1,\R).\nonumber
\end{eqnarray}

The Lie groups $\Gamma^\pm_{p,q}$, $\Gamma^+_{p,q}$ has the Lie algebra $\cl^0_{p,q}\oplus\cl^2_{p,q}$. All spin groups $\Pin(p,q)$, $\Spin(p,q)$, $\Pin_+(p,q)$, $\Pin_-(p,q)$, $\Spin_+(p,q)$ has the Lie algebra $\cl^2_{p,q}$.

Since Theorem \ref{thDC} and some facts from differential geometry, it follows that the spin groups are two-sheeted coverings of the orthogonal groups.

The groups $\Spin_+(p,q)$ are pathwise connected for $p\geq 2$ or $q\geq 2$. They are nontrivial covering groups of the corresponding orthogonal groups.

\begin{ex} The group $\Spin_+(1,1)=\{ue+ve_{12} \semicolon u^2-v^2=1\}$ is not pathwise connected (it is two branches of the hyperbole).
\end{ex}

The groups $\Spin_+(n)$, $n\geq 3$ and $\Spin_+(1,n-1)\cong\Spin_+(n-1,1)$, $n\geq 4$ are simply connected. They are the universal covering groups of the corresponding orthogonal groups.

\subsection{Other Lie Groups and Lie Algebras in Clifford Algebra}\label{sec:6.4}

Let us consider the following Lie groups and the corresponding Lie algebras (see Table \ref{tab:6}).

\begin{table}[ht]
\centering
\scriptsize{\begin{tabular}{|c|c|c|c|}\hline
& Lie group & Lie algebra & dimension\\ \hline
1& $(\C\otimes\cl_{p,q})^\times=\{U\in \C\otimes\cl_{p,q} \semicolon U^{-1} \,\mbox{exists}\}$ & $\overline{\textbf{0123}}\oplus \ii\overline{\textbf{0123}}$ &$2^{n+1}$  \\
2& $\cl^\times_{p,q}=\{U\in\cl_{p,q} \semicolon U^{-1}\, \mbox{exists}\}$ & $\overline{\textbf{0123}}$ &$2^n$  \\
3& $\cl^{(0) \times}_{p,q}=\{U\in \cl^{(0)}_{p,q} \semicolon U^{-1} \, \mbox{exists}\}$ & $\overline{\textbf{02}}$ &$2^{n-1}$  \\
4& $(\C\otimes\cl^{(0)}_{p,q})^\times=\{ U\in\C\otimes\cl_{p,q} \semicolon U^{-1}\, \mbox{exists}\}$ & $\overline{\textbf{02}}\oplus \ii\overline{\textbf{02}} $ &$2^n$  \\
5& $(\cl^{(0)}_{p,q}\oplus \ii \cl^{(1)}_{p,q})^\times=\{U\in \cl^{(0)}_{p,q}\oplus \ii \cl^{(1)}_{p,q} \semicolon  U^{-1}\,\mbox{exists}\}$ & $\overline{\textbf{02}}\oplus \ii\overline{\textbf{13}}$ &$2^n$  \\
6& $\G^{23\ii 01}_{p,q}=\{U\in \C\otimes\cl_{p,q} \semicolon \overline{\tilde{U}} U=e\}$ & $\overline{\textbf{23}}\oplus \ii\overline{\textbf{01}}$ & $2^n$  \\
7& $\G^{12\ii 03}_{p,q}=\{U\in \C\otimes\cl_{p,q} \semicolon \overline{\tilde{\hat{U}}} U=e\}$ & $\overline{\textbf{12}}\oplus \ii\overline{\textbf{03}}$ & $2^n$  \\
8& $\G^{2\ii 0}_{p,q}=\{U\in \cl^{(0)}_{p,q} \semicolon \overline{\tilde{U}} U=e\}$ & $\overline{\textbf{2}}\oplus \ii\overline{\textbf{0}}$ & $2^{n-1}$  \\
9& $\G^{23\ii 23}_{p,q}=\{U\in\C\otimes\cl_{p,q} \semicolon \tilde{U}U=e\}$ & $\overline{\textbf{23}}\oplus \ii\overline{\textbf{23}}$ & $2^{n}-2^{\frac{n+1}{2}}\sin{\frac{\pi (n+1)}{4}}$  \\
10& $\G^{12\ii 12}_{p,q}=\{U\in\C\otimes\cl_{p,q} \semicolon \hat{\tilde{U}}U=e\}$ & $\overline{\textbf{12}}\oplus \ii\overline{\textbf{12}}$ & $2^{n}-2^{\frac{n+1}{2}}\cos{\frac{\pi (n+1)}{4}}$  \\
11& $\G^{2\ii 2}_{p,q}=\{U\in\C\otimes\cl^{(0)}_{p,q} \semicolon \tilde{U}U=e\}$ & $\overline{\textbf{2}}\oplus \ii\overline{\textbf{2}}$ & $2^{n-1}-2^{\frac{n}{2}}\cos{\frac{\pi n}{4}}$  \\
12& $\G^{2\ii 1}_{p,q}=\{U\in\cl^{(0)}_{p,q}\oplus \ii\cl^{(1)}_{p,q} \semicolon \overline{\tilde{U}} U=e\}$ & $\overline{\textbf{2}}\oplus \ii\overline{\textbf{1}}$ & $2^{n-1}-2^{\frac{n-1}{2}}\cos{\frac{\pi (n+1)}{4}}$  \\
13& $\G^{2\ii 3}_{p,q}=\{U\in\cl^{(0)}_{p,q}\oplus \ii\cl^{(1)}_{p,q} \semicolon \overline{\tilde{\hat{U}}} U=e\}$ & $\overline{\textbf{2}}\oplus \ii\overline{\textbf{3}}$ & $2^{n-1}-2^{\frac{n-1}{2}}\sin{\frac{\pi (n+1)}{4}}$   \\
14& $\G^{23}_{p,q}=\{U\in\cl_{p,q} \semicolon \tilde{U}U=e\}$ & $\overline{\textbf{23}}$ & $2^{n-1}-2^{\frac{n-1}{2}}\sin{\frac{\pi (n+1)}{4}}$  \\
15& $\G^{12}_{p,q}=\{U\in\cl_{p,q} \semicolon \hat{\tilde{U}}U=e\}$ & $\overline{\textbf{12}}$ & $2^{n-1}-2^{\frac{n-1}{2}}\cos{\frac{\pi (n+1)}{4}}$  \\
16& $\G^{2}_{p,q}=\{U\in\cl^{(0)}_{p,q} \semicolon \tilde{U}U=e\}$ & $\overline{\textbf{2}}$ & $2^{n-2}-2^{\frac{n-2}{2}}\cos{\frac{\pi n}{4}}$\\  \hline
\end{tabular}}
\medskip
\caption{Lie groups and Lie algebras in $\C\otimes\cl_{p,q}$}\label{tab:6}
\end{table}

Isomorphisms for the group $\G^{23\ii 01}_{p,q}$ are proved in \cite{Snygg} (see also \cite{sh19}):
$$
\G^{23i01}_{p,q}\cong\left\lbrace
\begin{array}{ll}
\U(2^{\frac{n}{2}}), & \mbox{\rm if $p$ is even and $q=0$}\\
\U(2^{\frac{n-1}{2}})\oplus \U(2^{\frac{n-1}{2}}), & \mbox{\rm if $p$ is odd and $q=0$}\\
\U(2^{\frac{n-2}{2}},2^{\frac{n-2}{2}}), & \mbox{\rm if $n$ is even and $q\neq 0$}\\
\U(2^{\frac{n-3}{2}},2^{\frac{n-3}{2}})\oplus \U(2^{\frac{n-3}{2}},2^{\frac{n-3}{2}}), & \mbox{\rm if $p$ is odd and $q\neq 0$ is even}\\
\GL(2^{\frac{n-1}{2}}, \C), & \mbox{\rm if $p$ is even and $q$ is odd}
\end{array}
\right.\nonumber
$$

We call $\G^{23\ii 01}_{p,q}$ \emph{the pseudo-unitary group in Clifford algebra} and use it in some problems of the field theory \cite{msh3}, \cite{sh2}, \cite{msh4}.

Some of these Lie groups are considered in \cite{Port} and \cite{Lounesto}. Some of them are related to automorphism groups of the scalar products on the spinor spaces (\cite{Port}, \cite{Lounesto}, \cite{Benn:Tucker}, \cite{Abl3}).
Note that spin group $\Spin_+(p,q)$ is a subgroup of all groups in Table \ref{tab:6}. The group $\G^{2}_{p,q}$ coincides with $\Spin_+(p,q)$ in the cases $n\leq 5$. The Lie algebra of the spin group $\cl^2_{p,q}\in \overline{\textbf{2}}$ is a Lie subalgebra of all Lie algebras in Table \ref{tab:6}. We have $\cl^2_{p,q}=\overline{\textbf{2}}$ in the cases $n\leq 5$.

The isomorphisms for the group $\G^2_{p,q}$ are represented in Tables \ref{tab:7} and \ref{tab:8}.
There is $n \mod 8$ in the lines and $p-q \mod 8$ in the columns.

\begin{table}[ht]
\centering
\begin{tabular}{|l|l|l|l|}
\hline
$n \diagdown p-q$ & $0$ & $2, 6$  & $4$ \\ \hline
$0$  & $\begin{array}{l} ^2\O(2^{\frac{n-4}{2}},2^{\frac{n-4}{2}})\\ \mbox{if $p, q\neq 0$}\\ ^2\O(2^{\frac{n-2}{2}})\\ \mbox{if $p=0$ or $q=0$} \end{array}$ & $\O(2^{\frac{n-2}{2}},\C)$ & $^2\O(2^{\frac{n-4}{2}},\H)$  \\ \hline
$2, 6$ & $\GL(2^{\frac{n-2}{2}},\R)$& $\begin{array}{l} \U(2^{\frac{n-4}{2}},2^{\frac{n-4}{2}})\\ \mbox{if $p,q\neq 0$}\\ \U(\frac{n-2}{2})\\ \mbox{if $p=0$ or $q=0$} \end{array}$ &  $\GL(2^{\frac{n-4}{2}},\H)$ \\ \hline
$4$ & $^2\Sp(2^{\frac{n-4}{2}},\R)$  & $\Sp(2^{\frac{n-4}{2}},\C)$ &  $\begin{array}{l} ^2\Sp(2^{\frac{n-6}{2}},2^{\frac{n-6}{2}})\\ \mbox{if $p, q\neq 0$} \\ ^2\Sp(2^{\frac{n-4}{2}})\\ \mbox{if $p=0$ or $q=0$} \end{array}$\\ \hline
\end{tabular}
\medskip
\caption{Isomorphisms for the group $\G^{2}_{p,q}$ in the cases of even $n$}\label{tab:7}
\end{table}
\begin{table}[ht]
\centering
\begin{tabular}{|l|l|l|}
\hline
$n \diagdown p-q$ & $1, 7$  & $3, 5$ \\ \hline
$1, 7$  & $\begin{array}{l} \O(2^{\frac{n-3}{2}},2^{\frac{n-3}{2}})\\ \mbox{if $p, q\neq 0$}\\ \O(2^{\frac{n-1}{2}})\\ \mbox{if $p=0$ or $q=0$}\end{array}$ &  $\O(2^{\frac{n-3}{2}},\H)$   \\ \hline
$3, 5$ & $\Sp(2^{\frac{n-3}{2}},\R)$ &  $\begin{array}{l} \Sp(2^{\frac{n-5}{2}},2^{\frac{n-5}{2}})\\ \mbox{if $p, q\neq 0$}\\ \Sp(2^\frac{n-3}{2})\\ \mbox{if $p=0$ or $q=0$}. \end{array}$  \\ \hline
\end{tabular}
\medskip
\caption{Isomorphisms for the group $\G^{2}_{p,q}$ in the cases of odd $n$}\label{tab:8}
\end{table}

One can find isomorphisms for all remaining Lie groups and corresponding Lie algebras from Table \ref{tab:6} in a series of papers \cite{sh13}, \cite{sh16}, and \cite{sh19}.

\section{Dirac Equation and Spinors in n Dimensions}\label{sec:7}

\subsection{Dirac Equation in Matrix Formalism}\label{sec:7.1}

In Section \ref{sec:7}, we use the notation with upper indices for the Dirac gamma-matrices and the generators of the Clifford algebra because of the useful covariant form of the Dirac equation.

Let $\R^{1,3}$ be Minkowski space with Cartesian coordinates $x^\mu$, $\mu=0, 1, 2, 3$. The metric tensor of Minkowski space is given by a diagonal matrix
$$\eta=\diag(1, -1, -1, -1).$$
We denote partial derivatives by $\partial_\mu:=\frac{\partial}{\partial x^\mu}$.

The \emph{Dirac equation} for the electron \cite{Dirac}, \cite{Dirac2} can be written in the following way
\begin{eqnarray}
\ii\gamma^\mu (\partial_\mu \psi-\ii a_\mu \psi)-m\psi=0\nonumber
\end{eqnarray}
where $a_\mu:\R^{1,3}\to \R$ is the electromagnetic 4-vector potential, $m\geq 0\in\R$ is the electron mass,
$\psi:\R^{1,3}\to \C^4$ is the wave function (the Dirac spinor) and $\gamma^\mu$ are the Dirac gamma-matrices which satisfy conditions
$$\gamma^\mu\gamma^\nu+\gamma^\nu\gamma^\mu=2\eta^{\mu\nu}{\bf 1},\quad \gamma^\mu\in\Mat(4,\C).$$

The Dirac equation is gauge invariant. If we take the expressions
$$a_\mu \to \acute a_\mu=a_\mu +\lambda(x),\quad \psi\to\acute\psi=\psi \e^{\ii\lambda(x)},\qquad \lambda(x)\in\R$$
then they satisfy the same equation:
$$\ii\gamma^\mu (\partial_\mu \acute\psi -\ii \acute a_\mu\acute \psi)-m\acute\psi=
\ii \gamma^\mu (\partial_\mu (\e^{i\lambda}\psi) -i(a_\mu+\partial_\mu \lambda)(\e^{\ii\lambda}\psi))-m(\e^{\ii\lambda}\psi)$$
$$=\ii \gamma^\mu (\ii(\partial_\mu\lambda)\e^{\ii\lambda}\psi+\e^{\ii\lambda}(\partial_\mu\psi) -\ii a_\mu \e^{\ii\lambda}\psi-\ii(\partial_\mu \lambda) \e^{\ii\lambda}\psi)-m\e^{\ii\lambda}\psi$$
$$=\e^{\ii\lambda}(\ii \gamma^\mu (\partial_\mu \psi -\ii a_\mu\psi)-m\psi)=0.$$

One says that the Dirac equation is gauge invariant with respect to the gauge group
$$\U(1)=\{ \e^{\ii\lambda} \semicolon \lambda\in\R\}.$$
The corresponding Lie algebra is
$$\u(1)=\{\ii\lambda \semicolon \lambda\in\R\}.$$

The Dirac equation is relativistic invariant. Let us consider orthogonal transformation of coordinates
$$x^\mu\to \acute x^{\mu}=p^\mu_\nu x^\nu,\quad P=||p^\mu_\nu||\in\O(1,3).$$
Then
$$\partial_\mu \to \acute \partial_\mu=q_\mu^\nu \partial_\nu,\quad a_\mu \to \acute a_\mu=q_\mu^\nu a_\nu,\quad Q=||q^\mu_\nu||=P^{-1}.$$

There are two points of view on transformations of the Dirac gamma-matrices and the wave function (see \cite{Som}).

In the first (tensor) approach, we have
$$\gamma^\mu \to \acute \gamma^{\mu}=p_\nu^\mu \gamma^\nu,\quad \psi\to \acute \psi=\psi.$$
In this approach, all expressions are tensors and the Dirac equation is relativistic invariant. The tensor approach is considered in details in \cite{marchuk}.

In the second (spinor) approach, we have
$$\gamma^\mu \to \acute \gamma^{\mu}=\gamma^\mu,\quad \psi \to \acute\psi=S\psi,\quad S^{-1}\gamma^\mu S=p^\mu_\nu \gamma^\nu$$
$$\ii\acute\gamma^{\mu}(\acute\partial_\mu\acute\psi-\ii \acute a_\mu \acute\psi)-m\acute\psi=\ii \gamma^\mu(q^\nu_\mu \partial_\nu (S \psi)-\ii q^\nu_\mu a_\nu S \psi)-m S\psi)$$
$$=S(\ii S^{-1}q^\nu_\mu\gamma^\mu S(\partial_\nu\psi-\ii a_\mu \psi)-m\psi)=S(\ii\gamma^\nu(\partial_\nu\psi-\ii a_\mu \psi)-m\psi)=0.$$
In this approach, the Dirac gamma-matrices do not change and the wave function $\psi$ changes as spinor with the aid of the element $S$ of the spin group. The formula $S^{-1}\gamma^\mu S=p^\mu_\nu \gamma^\nu$ describes the double cover of the orthogonal group by the spin group. This approach is generally accepted.

\subsection{Dirac Equation in Formalism of Clifford Algebra}\label{sec:7.2}

Let us consider the complexified Clifford algebra $\C\otimes\cl_{1,3}$ with the generators $e^0, e^1, e^2, e^3$. In Section \ref{sec:7}, we use notation with upper indices for the generators of the Clifford algebra.

We have a primitive idempotent
$$t=\frac{1}{2}(e+e^0)\frac{1}{2}(e+\ii e^{12})\leftrightarrow \left(
                        \begin{array}{cccc}
                          1 & 0 & 0 & 0 \\
                          0 & 0 & 0 & 0 \\
                          0 & 0 & 0 & 0 \\
                          0 & 0 & 0 & 0 \\
                        \end{array}
                      \right),\qquad t^2=t=t^\dagger.
$$
The Dirac spinor is
$$\psi\leftrightarrow \left(
                        \begin{array}{cccc}
                          \psi_1 & 0 & 0 & 0 \\
                          \psi_2 & 0 & 0 & 0 \\
                          \psi_3 & 0 & 0 & 0 \\
                          \psi_4 & 0 & 0 & 0 \\
                        \end{array}
                      \right)
\in \I(t)=(\C\otimes\cl_{1,3})t.
$$
The corresponding left ideal $\I(t)$ is called \emph{spinor space}.

The Dirac equation can be written in the following form
$$
\ii e^\mu (\partial_\mu \psi-\ii a_\mu \psi)-m\psi=0
$$
where $\psi$ is an element of the left ideal of the Clifford algebra.

All properties of the Dirac equation from the previous section are valid.

\subsection{Dirac-Hestenes Equation}\label{sec:7.3}

Let us consider Minkowski space $\R^{1,3}$ and the complexified Clifford algebra $\C\otimes\cl_{1,3}$ with the generators $e^0, e^1, e^2, e^3$. We have a primitive idemptonent $t=\frac{1}{4}(e+E)(e-\ii I)$ and the corresponding left ideal $\I(t)$, where
$E:=e^0$, $I:=-e^{12}$, $t^2=t=t^\dagger$, $\ii t=It$, $t=Et$.

\begin{lemma}\label{lemmaDH} For arbitrary $U\in \I(t)$ the equation $X t=U$ has a unique solution $X\in\cl^{(0)}_{1,3}$ (and a unique solution $X\in\cl^{(1)}_{1,3}$).
\end{lemma}
\begin{proof} We can choose the orthonormal basis of the left ideal $\I(t)$ of the following form:
$$\tau_k=F_k t,\quad k=1, 2, 3, 4,\quad F_1=2e, F_2=2e^{13}, F_3=2e^{03}, F_4=2e^{01}\in\cl^{0)}_{1,3}.$$
We have $U=(\alpha^k+\ii \beta^k)\tau_k$ for some $\alpha^k, \beta^k\in \R$.

Using $\ii t=It$, we conclude that $X=F_k(\alpha^k+I\beta^k)\in\cl^{(0)}_{1,3}$ is a solution of $Xt=U$.

Now let us prove the following statement. If the element $Y\in\cl^{(0)}_{1,3}$ is a solution of equation $Yt=0$, then $Y=0$. For element $Yt\in \I(t)$ we have
$$Y t=\frac{1}{2}((y-\ii y_{12})\tau_1+(-y_{13}-\ii y_{23})\tau_2+(y_{03}-\ii y_{0123})\tau_3+ (y_{01}+\ii y_{02})\tau_4)=0.$$

Using $t=Et$, we conclude that $X=F_k E(\alpha^k+I\beta^k)\in\cl^{(1)}_{1,3}$ is also a solution of equation $Xt=U$. The proof of uniqueness in this case is similar.
\end{proof}
One can find this lemma and similar statements, for example, in \cite{marchuk}.

Let us rewrite the Dirac equation
$\ii e^\mu(\partial_\mu \psi-\ii a_\mu\psi)-m\psi=0$ in the following form
\begin{eqnarray}
e^\mu(\partial_\mu\psi-\ii a_\mu\psi)+\ii m\psi=0,\quad \psi\in \I(t).\label{dir1}
\end{eqnarray}

The \emph{Dirac-Hestenes equation} \cite{Hestenes} is
\begin{eqnarray}
e^\mu (\partial_\mu\Psi-a_\mu\Psi I)E+m\Psi I=0,\quad \Psi\in\cl^{(0)}_{1,3}.\label{dir2}
\end{eqnarray}

\begin{theorem} The Dirac equation and the Dirac-Hestenes equation are equivalent.
\end{theorem}
\begin{proof} Let us multiply both sides of the Dirac-Hestenes equation (\ref{dir2}) by $t$ on the right. Using $Et=t$, $It=\ii t$, and $\Psi t=\psi$, we obtain the Dirac equation (\ref{dir1}).

Now let us start with the Dirac equation (\ref{dir1}). We have $\psi\in\I(t)$. Using Lemma \ref{lemmaDH}, we conclude that there exists $\Psi\in\cl^{(0)}_{1,3}$ such that $\Psi t=\psi$. Using $Et=t$ and $It=\ii t$, we obtain $$\underbrace{(e^\mu(\partial_\mu\Psi-a_\mu\Psi I)E+m\Psi I)}_{\in\cl^{(0)}_{1,3}}t=0.$$
Using Lemma \ref{lemmaDH} for the second time, we obtain the Dirac-Hestenes equation (\ref{dir2}).
\end{proof}

Note that the dimensions of the spinor spaces are the same in two approaches:
$$\dim \I(t)=\dim \C^4=8,\qquad \dim \cl^{(0)}_{1,3}=8.$$

The Dirac-Hestenes equation is widely used in applications (see, for example, \cite{rodr}, \cite{lasenby}).

\subsection{Weyl, Majorana and Majorana-Weyl Spinors}\label{sec:7.4}

Detailed information on $n$-dimensional spinors (using the methods of Clifford algebra) can be found in \cite{sh12}. See also \cite{Benn:Tucker}.

We study the connection between matrix operations (transpose, matrix complex conjugation) and operations in Clifford algebra (reverse, complex conjugation, grade involution), we introduce the notion of additional signature of the Clifford algebra (for more details, see \cite{sh12}, also \cite{sh13}, \cite{sh16}, \cite{sh19}, where we develop and use these results).

Let us consider chirality operator (pseudoscalar) in $\C\otimes\cl_{p,q}$:
$$\omega=\left\{
                                             \begin{array}{ll}
                                               e^{1\ldots n}, & \hbox{$p-q=0, 1\mod 4$} \\
                                               \ii e^{1\ldots n}, & \hbox{$p-q=2, 3\mod 4$.}
                                             \end{array}
                                           \right.
$$
We have
$$\omega=\omega^{-1}=\omega^\dagger.$$

Let us consider orthogonal idempotents
$$P_L:=\frac{1}{2}(e-\omega),\quad P_R:=\frac{1}{2}(e+\omega)$$
$$P_L^2=P_L,\qquad P_R^2=P_R,\qquad P_L P_R=P_R P_L=0.$$

In the case of odd $n$, the complexified Clifford algebra $\C\otimes\cl_{p,q}$ is the direct sum of two ideals:
$$
\C\otimes\cl_{p,q}=P_L (\C\otimes\cl_{p,q})\oplus P_R (\C\otimes\cl_{p,q}),\quad \C\otimes\cl_{p,q} \cong {}^2\Mat(2^{\frac{n-1}{2}},\C).
$$

Let us consider the case of even $n$. For the set of Dirac spinors $E_D=\{\psi\in \I(t)\}$ we have
$$E_D=E_{LW}\oplus E_{RW}$$
where
$$E_{LW}:=\{\psi\in E_D \semicolon P_L \psi=\psi \}=\{\psi\in E_D \semicolon \omega \psi=-\psi \}$$
is the set of left Weyl spinors and
$$E_{RW}:=\{\psi\in E_D \semicolon P_R \psi=\psi\}=\{\psi\in E_D \semicolon \omega \psi=\psi\}$$
is the set of right Weyl spinors.

Using the Pauli theorem (Theorems \ref{thPauli1} and \ref{thPauli2}), we obtain existence of the elements $A_\pm$ such that:
\begin{eqnarray}
(e^a)^\dagger=\pm A_{\pm}^{-1} e^a A_{\pm}.\label{tA}
\end{eqnarray}
If $n$ is even, then both elements $A_{\pm}$ exist. If $p$ is odd and $q$ is even, then only $A_+$ exists. If $p$ is even and $q$ is odd, then only $A_-$ exists.

We can rewrite (\ref{tA}) in the following way:
$$U^\dagger=A_{+}^{-1} \overline{\tilde{U}} A_+,\quad U^\dagger=A_{-}^{-1} \widehat{\overline{\tilde{U}}} A_-,\quad U\in\C\otimes\cl_{p,q}.$$

The explicit formulas for $A_\pm$ are given in Theorem \ref{thHerm}.

Let us consider two \emph{Dirac conjugations}
$$\psi^{D_\pm}:=\psi^\dagger (A_{\pm})^{-1}.$$

\begin{ex} In the case $(p,q)=(1,3)$ with the gamma-matrices $\gamma^0$, $\gamma^1$, $\gamma^2$, $\gamma^3$, we obtain the standard Dirac conjugation $\psi^{D_+}=\psi^\dagger \gamma^0$ and one else $\psi^{D-}=\psi^\dagger \gamma^{123}$.
\end{ex}

The Dirac conjugation is used to define \emph{bilinear covariants}
$$j_{\pm}^A=\psi^{D_\pm}e^A \psi.$$
The \emph{Dirac current} $\psi^{D_+} e^\mu \psi$ is a particular case of the bilinear covariants. Using the Dirac equation, it is not difficult to obtain the law of conservation of the Dirac current:
$$\partial_\mu(\psi^{D_+} e^\mu \psi)=0.$$

We denote the matrix complex conjugation by $\mcc{\,}$. It should not be confused with the operation of complex conjugation in the complexified Clifford algebra $\C\otimes\cl_{p,q}$.

Let us consider the following two operations in $\C\otimes\cl_{p,q}$:
$$U^\T:=\beta^{-1} (\beta^\T(U)),\qquad \mcc{U}:=\beta^{-1}(\mcc{\beta(U)}),\qquad U\in\C\otimes\cl_{p,q}$$
where
$$\beta: \C\otimes\cl_{p,q} \to
\left\lbrace
\begin{array}{ll}
\Mat(2^{\frac{n}{2}}, \C), & \parbox{.5\linewidth}{if $n$ is even}\\
\Mat(2^{\frac{n-1}{2}}, \C)\oplus \Mat(2^{\frac{n-1}{2}}, \C), & \parbox{.5\linewidth}{if $n$ is odd}
\end{array}\nonumber
\right.
$$
is the faithful representation of $\C\otimes\cl_{p,q}$ of the minimal dimension. These two operations depend on the representation $\beta$.

Using the Pauli theorem (Theorems \ref{thPauli1} and \ref{thPauli2}), we obtain existence of the elements $C_\pm$ such that:
\begin{eqnarray}
(e^a)^\T=\pm C_{\pm}^{-1} e^a C_{\pm}.\label{tC}
\end{eqnarray}
If $n$ is even, then both elements $C_{\pm}$ exist. If $n=1\mod 4$, then only $C_+$ exists. If $n=3\mod 4$, then only $C_-$ exists.

We can rewrite (\ref{tC}) in the following way:
$$U^\T=C_{+}^{-1} \tilde{U} C_+,\quad U^\T=C_{-}^{-1} \widehat{\tilde{U}} C_-,\quad U\in\C\otimes\cl_{p,q}.$$

The explicit formulas for $C_\pm$ are given in \cite{sh12} using the notion of additional signature of the Clifford algebra. Also these elements have the following properties:
$$(C_{\pm})^\T=\lambda_{\pm} C_{\pm},\qquad \mcc{C_{\pm}}C_{\pm}=\lambda_{\pm} e$$
$$\lambda_+=\left\lbrace
\begin{array}{ll}
+1, & \mbox{$n\equiv 0, 1, 2\!\!\mod 8$}\\
-1, &  \mbox{$n\equiv 4, 5, 6\!\!\mod 8$}
\end{array}
\right.\quad
\lambda_-=\left\lbrace
\begin{array}{ll}
+1, & \mbox{$n\equiv 0, 6, 7\!\!\mod 8$}\\
-1, &  \mbox{$n\equiv 2, 3, 4\!\!\mod 8$}.
\end{array}
\right.
$$

Using the Pauli theorem (Theorems \ref{thPauli1} and \ref{thPauli2}), we obtain existence of the elements $B_\pm$ such that:
\begin{eqnarray}
\mcc{e^a}=\pm B_{\pm}^{-1} e^a B_{\pm}.\label{tB}
\end{eqnarray}
If $n$ is even, then both elements $B_{\pm}$ exist. If $p-q=1\mod 4$, then only $B_+$ exists. If $p-q=3\mod 4$, then only $B_-$ exists.

We can rewrite (\ref{tB}) in the following way:
$$\mcc{U}=B_{+}^{-1} \overline{U} B_+,\quad \mcc{U}=B_{-}^{-1} \widehat{\overline{U}} B_-,\quad U\in\C\otimes\cl_{p,q}.$$

The explicit formulas for $B_\pm$ are given in \cite{sh12} using the notion of additional signature of the Clifford algebra. Also these elements have the following properties:
$$B_{\pm}^T=\epsilon_{\pm} B_{\pm},\qquad \mcc{B_{\pm}} B_{\pm}=\epsilon_{\pm} e$$
$$\epsilon_+=\left\lbrace
\begin{array}{ll}
+1, & \mbox{$p-q\equiv 0, 1, 2\!\!\mod 8$}\\
-1, &  \mbox{$p-q\equiv 4, 5, 6\!\!\mod 8$}
\end{array}
\right.\,
\epsilon_-=\left\lbrace
\begin{array}{ll}
+1, & \mbox{$p-q\equiv 0, 6, 7\!\!\mod 8$}\\
-1, &  \mbox{$p-q\equiv 2, 3, 4\!\!\mod 8$}.
\end{array}
\right.
$$

We introduce the \emph{Majorana conjugation} in the following way
$$\psi^{M_\pm}:=\psi^\T (C_{\pm})^{-1}.$$

\begin{ex} In the case $(p,q)=(1,3)$, we have $\psi^{M_+}=\psi^\dagger (\gamma^{13})^{-1}$ and $\psi^{M_-}=\psi^\dagger (\gamma^{02})^{-1}$.
\end{ex}

We introduce the \emph{charge conjugation} in the following way
$$\psi^{ch_\pm}:=B_{\pm}\mcc{\psi}.$$

\begin{ex} In the case $(p,q)=(1,3)$, we have $\psi^{ch_+}=\gamma^{013}\mcc{\psi}$ and $\psi^{ch_-}=\gamma^2\mcc{\psi}$.
\end{ex}

We have the following relation between $A_\pm$, $B_\pm$, and $C_\pm$ (when they exist):
$$B_{+}=\widetilde{A_{+}^{-1}} C_{+},\quad B_{+}=\widehat{\widetilde{A_{-}^{-1}}} C_{-},\quad B_{-}=\widetilde{A_{-}^{-1}} C_{+},\quad B_{-}=\widehat{\widetilde{A_{+}^{-1}}} C_{-}$$
$$\psi^{ch_+}=C_+ (\psi^{D_+})^\T=C_- (\psi^{D_-})^\T,\qquad \psi^{ch_-}=C_- (\psi^{D_+})^\T=C_+ (\psi^{D_-})^\T.$$

Let us denote the set of \emph{Majorana spinors} by
$$E_M:=\{ \psi\in E_D \semicolon \psi^{ch_-}=\pm\psi\}$$
and the set of \emph{pseudo-Majorana spinors} by
$$E_{psM}:=\{ \psi\in E_D \semicolon \psi^{ch_+}=\pm\psi\}.$$

Using definition of the charge conjugation and the properties of $B_{\pm}$, it can be proved that Majorana spinors are realized only in the cases  $p-q=0, 6, 7\mod 8$ and pseudo-Majorana spinors are realized only in the cases $p-q=0, 1, 2\mod 8$ (see, for example, \cite{sh12}).

Let us denote a set of \emph{left Majorana-Weyl spinors} by
$$E_{LMW}:=\{ \psi\in E_{LW} \semicolon \psi^{ch_-}=\pm\psi\}=\{ \psi\in E_{LW} \semicolon \psi^{ch_+}=\pm\psi\}$$
and a set of \emph{right Majorana-Weyl spinors} by
$$E_{RMW}:=\{ \psi\in E_{RW} \semicolon \psi^{ch_-}=\pm\psi\}=\{ \psi\in E_{RW} \semicolon \psi^{ch_+}=\pm\psi\}.$$
It can be proved that Majorana-Weyl spinors are realized only in the cases $p-q=0\mod 8$ (see, for example, \cite{sh12}).

The question of existence of the spinors in the cases of different dimensions and signatures is related to the supersymmetry theory (see classic works on supersymmetry and supergravity \cite{sup1}, \cite{sup2} and other papers and reviews \cite{sup3}, \cite{sup4}, \cite{sup5}, \cite{sup6}, \cite{sup7}, \cite{sup8}, \cite{sup9}, \cite{sup10}).

\section*{Acknowledgements}

The author is grateful to Professor Iva\"ilo~M.~Mladenov for invitation and support. The author is grateful to Professor Nikolay~G.~Marchuk for fruitful discussions.

The reported study was funded by RFBR according to the research project
No. 16-31-00347 mol\_a.

\printindex
\end{document}